\documentclass{article}

\usepackage{cmbright}
\usepackage{amssymb,amsmath,amsthm}
\usepackage{mathrsfs}
\usepackage[T1]{fontenc}

\def\A{\mathcal A}
\def\B{\mathscr B}
\def\C{\mathbb C}
\def\d{\mathrm{d}}
\def\D{\mathbb D}
\def\DD{\mathscr D}
\def\F{\mathscr F}
\def\G{\mathcal G}
\def\H{\mathcal H}
\def\HH{\mathbb H}
\def\K{\mathscr K}
\def\KK{\mathcal K}
\def\M{\mathsf M}
\def\N{\mathbb N}

\def\R{\mathbb R}

\def\T{\mathbb T}
\def\U{\mathsf U}
\def\Z{\mathbb Z}

\def\dom{\mathcal D}
\def\supp{\mathop{\mathrm{supp}}\nolimits}

\def\e{\mathop{\mathrm{e}}\nolimits}
\def\im{\mathop{\mathrm{Im}}\nolimits}
\def\re{\mathop{\mathrm{Re}}\nolimits}
\def\Int{\mathop{\mathrm{Int}}\nolimits}
\def\ltwo{\mathop{\mathrm{L}^2}\nolimits}
\def\linf{\mathop{\mathrm{L}^\infty}\nolimits}
\DeclareMathOperator*{\slim}{s\hspace{0.1pt}-\hspace{0.1pt}lim}

\newtheorem{Theorem}{Theorem}[section]

\newtheorem{Lemma}[Theorem]{Lemma}
\newtheorem{Corollary}[Theorem]{Corollary}
\newtheorem{Proposition}[Theorem]{Proposition}

\newtheorem{Assumption}[Theorem]{Assumption}

\begin{document}

%--------------------------------------------------------------------------------------
% Title
%--------------------------------------------------------------------------------------

\title{Quantum walks with an anisotropic coin I\;\!: spectral theory}

\author{S. Richard$^1$\footnote{Supported by JSPS Grant-in-Aid for Young Scientists A
no 26707005, and on leave of absence from Univ.~Lyon, Universit\'e Claude Bernard Lyon
1, CNRS UMR 5208, Institut Camille Jordan, 43 blvd.~du 11 novembre 1918, F-69622
Villeurbanne cedex, France},
A. Suzuki$^2$\footnote{Supported by JSPS Grant-in-Aid for
Young Scientists B
no 	26800054},
R. Tiedra de Aldecoa$^3$\footnote{Supported by the Chilean Fondecyt Grant 1130168.}}

\date{\small}
\maketitle \vspace{-1cm}

\begin{quote}
\emph{
\begin{itemize}
\item[$^1$] Graduate school of mathematics, Nagoya University,
Chikusa-ku, \\
Nagoya 464-8602, Japan
\item[$^2$] Division of Mathematics and Physics, Faculty of Engineering, Shinshu University, \\
Wakasato, Nagano 380-8553, Japan
\item[$^3$] Facultad de Matem\'aticas, Pontificia Universidad Cat\'olica de Chile,\\
Av. Vicu\~na Mackenna 4860, Santiago, Chile
\item[] \emph{E-mails:} richard@math.nagoya-u.ac.jp,
akito@shinshu-u.ac.jp,
rtiedra@mat.puc.cl
\end{itemize}
}
\end{quote}

%--------------------------------------------------------------------------------------
% Abstract
%--------------------------------------------------------------------------------------

\begin{abstract}
We perform the spectral analysis of the evolution operator $U$ of quantum walks with
an anisotropic coin, which include one-defect models, two-phase quantum walks, and
topological phase quantum walks as special cases. In particular, we determine the
essential spectrum of $U$, we show the existence of locally $U$-smooth operators, we
prove the discreteness of the eigenvalues of $U$ outside the thresholds, and we prove
the absence of singular continuous spectrum for $U$. Our analysis is based on new
commutator methods for unitary operators in a two-Hilbert spaces setting, which are of
independent interest.
\end{abstract}


\smallskip


%--------------------------------------------------------------------------------------
%\tableofcontents
%--------------------------------------------------------------------------------------

%--------------------------------------------------------------------------------------
\section{Introduction}\label{Sec_Introduction}
\setcounter{equation}{0}
%--------------------------------------------------------------------------------------

The notion of discrete-time quantum walks appears in numerous contexts 
\cite{ADZ93,ABNVW01,GZ88,Gud88,Mey96,Wat01}. Among them, Gudder \cite{Gud88}, Meyer
\cite{Mey96}, and Ambainis et al. \cite{ABNVW01} introduced one-dimensional quantum
walks as a quantum mechanical counterpart of classical random walks. Nowadays, these
quantum walks and their generalisations have been physically implemented in various
ways \cite{MW14}. Versatile applications of quantum walks can be found in
\cite{COB15,IMSY15,Por13,Ven12} and references therein.

Recently, because of the controllability of their parameters, discrete-time quantum
walks have attracted attention as promising candidates to realise topological
insulators. In a series of papers \cite{KBFRBKADW12,KRBD10}, Kitagawa et al. have
shown that one and two dimensional quantum walks possess topological phases, and they
experimentally observed a topologically protected bound state between two distinct
phases. See \cite{Kit12} for an introductory review on the topological phenomena in
quantum walks. Motivated by these studies, Endo et al. \cite{EKH15} (see also
\cite{EEKST15,EEKST16}) have performed a thorough analysis of the asymptotic behaviour
of two-phase quantum walks, whose evolution is given by a unitary operator
$U_{\rm TP}=SC$ with $S$ a shift operator and $C$ a coin operator defined as a
multiplication by unitary matrices $C(x)\in\U(2)$, $x\in\Z$. When $C(x)$ is given by
\begin{equation} \label{eq_two_phase}
C(x)=
\begin{cases}
\frac1{\sqrt2}
\begin{pmatrix}
1 & \e^{i\sigma_+}\\
\e^{-i\sigma_+} & -1
\end{pmatrix}
& \hbox{if $x\ge0$}\bigskip\\
\frac1{\sqrt2}
\begin{pmatrix}
1 & \e^{i\sigma_-}\\
\e^{-i\sigma_-} & -1
\end{pmatrix}
& \hbox{if $x\le-1$}
\end{cases}
\end{equation}
with $\sigma_\pm\in[0,2\pi)$, the two-phase quantum walk with evolution operator
$U_{\rm TP}$ is called complete two-phase quantum walk, and when $C(x)$ satisfies the
alternative condition at $0$
\begin{equation}\label{eq_defect}
C(0)=\begin{pmatrix} 1 & 0 \\ 0 & -1 \end{pmatrix},
\end{equation}
the quantum walk is called two-phase quantum walk with one defect. In
\cite{EEKST16,EKH15}, Endo et al. have proved a weak limit theorem \cite{Kon02,Kon05}
similar to the de Moivre-Laplace theorem (or the Central limit theorem) for random
walks, which describes the asymptotic behaviours of the two-phase quantum walk.

In the present paper and the companion paper \cite{RST_2}, we consider one-dimensional
quantum walks $U=SC$ with a coin operator $C$ exhibiting an anisotropic behaviour at
infinity, with short-range convergence to the asymptotics. Namely, we assume that
there exist matrices $C_{\ell},C_{\rm r}\in\U(2)$ and positive constants
$\varepsilon_\ell,\varepsilon_{\rm r}>0$ such that
\begin{equation}\label{eq_aniso_coin}
C(x)=
\begin{cases}
C_{\ell}+O\big(|x|^{-1-{\varepsilon_\ell}}\big) & \hbox{as $x\to-\infty$}\\
C_{\rm r}+O\big(|x|^{-1-{\varepsilon_r}}\big) & \hbox{as $x\to\infty$}.
\end{cases}
\end{equation}
We call this type of quantum walks quantum walks with an anisotropic coin or simply
anisotropic quantum walks. They include two-phase quantum walks with coins defined by
\eqref{eq_two_phase} and \eqref{eq_defect} and one-defect models
\cite{CGMV12,Kon10,KLS13,WLKGGB12} as special cases. In the case
$C_0:=C_{\ell}=C_{\rm r}$ and $\varepsilon_0:=\varepsilon_{\ell}=\varepsilon_{\rm r}$,
quantum walks with an anisotropic coin reduce to one-dimensional quantum walks with a
position dependent coin
$$
C(x)=C_0+O\big(|x|^{-1-\varepsilon_0}\big),\quad |x|\to\infty,
$$
for which the absence of the singular continuous spectrum was proved in \cite{ABJ15}
and for which a weak limit theorem was derived in \cite{Suz16}.

Quantum walks with an anisotropic coin are also related to Kitagawa's topological
quantum walk model called a split-step quantum walk \cite{Kit12,KBFRBKADW12,KRBD10}.
Indeed, if $R(\theta)\in\U(2)$ is a rotation matrix with rotation angle $\theta/2$,
$R(\Theta_j)$ the multiplication operator by
$R\big(\theta_j(\;\!\cdot\;\!)\big)\in\U(2)$ with $\theta_j:\Z\to[0,2\pi)$, $j=1,2$,
and $T_\downarrow,T_\uparrow$ shift operators satisfying
$S=T_\downarrow T_\uparrow=T_\uparrow T_\downarrow$, then the evolution operator of
the split-step quantum walk is defined as
$$
U_{\rm SS}(\theta_1,\theta_2)
:=T_\downarrow\;\!R(\Theta_2)\;\!T_\uparrow\;\!R(\Theta_1).
$$
Now, as mentioned in \cite{Kit12}, $U_{\rm SS}(\theta_1,\theta_2)$ is unitarily
equivalent to $T_\uparrow R(\Theta_1) T_\downarrow R(\Theta_2)$. Thus, our evolution
operator $U$ describes a quantum walk unitarily equivalent to the one described by
$U_{\rm SS}(\theta_1,\theta_2)$ if $\theta_1 \equiv 0$ and
$C(\;\!\cdot\;\!)=R\big(\theta_2(\;\!\cdot\;\!)\big)$ (see \cite{Ohn16,SS16} for the
definition of unitary equivalence between two quantum walks). In \cite{Kit12},
Kitagawa dealt with the case
$$
\theta_2(x)
:=\frac12(\theta_{2-}+\theta_{2+})+\frac12(\theta_{2+}-\theta_{2-})\tanh(x/3),
\quad \theta_{2-},\theta_{2+}\in[0,2\pi),~x\in\Z,
$$
which corresponds to taking the anisotropic coin \eqref{eq_aniso_coin} with
$C_{\ell}=R(\theta_{2-})$ and $C_{\rm r}=R(\theta_{2+})$, and which cannot be covered
by two-phase models.

The main goal of the present paper and \cite{RST_2} is to establish a weak limit
theorem for the the evolution operator $U$ of the quantum walk with an anisotropic
coin satisfying \eqref{eq_aniso_coin}. As put into evidence in \cite{Suz16}, in order
to establish a weak limit theorem one has to prove along the way the following two
important results:
\begin{itemize}
\item[(i)] the absence of singular continuous spectrum,
\item[(ii)] the existence of the asymptotic velocity.
\end{itemize}

In the present paper, we perform the spectral analysis of the evolution operator $U$
of quantum walks with an anisotropic coin. We determine the essential spectrum of $U$,
we show the existence of locally $U$-smooth operators, we prove the discreteness of
the eigenvalues of $U$ outside the thresholds, and we prove the absence of singular
continuous spectrum for $U$. In the companion paper \cite{RST_2}, we will develop the
scattering theory for the evolution operator $U$. We will prove the existence and the
completeness of wave operators for $U$ and a free evolution operator $U_0$, we will
show the existence of the asymptotic velocity for $U$, and we will finally establish a
weak limit theorem for $U$. Other interesting related topics such as the existence and
the robustness of a bound state localised around the phase boundary or a weak limit
theorem for the split-step quantum walk with $\theta_1\ne0$ are considered in
\cite{FFS1loc} and \cite{FFS1wlt}, respectively.

The rest of this paper is structured as follows. In Section 2, we give the precise
definition of the evolution operator $U$ for the quantum walk with an anisotropic coin
and we state our main results on the essential spectrum of $U$ (Theorem
\ref{thm_essential}), the locally $U$-smooth operators (Theorem
\ref{thm_U_smooth_walks}), and the eigen\-values and singular continuous spectrum of
$U$ (Theorem \ref{thm_spectrum_U_walks}). Section 3 is devoted to mathematical
preliminaries. Here we develop new commutator methods for unitary operators in a
two-Hilbert spaces setting, which are a key ingredient for our analysis and are of
independent interest. In Section 4, we prove our main theorems as an application of
the commutator methods developed in Section 3. In Subsection \ref{Sec_Mourre}, we
prove Theorem \ref{thm_essential} and we define in Lemma \ref{lemma_A_self} a
conjugate operator $A$ for the evolution operator $U$ built from conjugate operators
for the asymptotic evolution operators $U_\ell:=SC_\ell$ and $U_{\rm r}:=SC_{\rm r}$,
where $C_{\ell}$ and $C_{\rm r}$ are the constant coin matrices given in
\eqref{eq_aniso_coin}. Finally, in Subsection \ref{Sec_spectrum} we prove Theorems
\ref{thm_U_smooth_walks} and \ref{thm_spectrum_U_walks}.\\

\noindent
{\bf Acknowledgements.} The third author thanks the Graduate School of Mathematics of
Nagoya University for its warm hospitality in January-February 2017.

%--------------------------------------------------------------------------------------
\section{Model and main results}\label{Sec_model}
\setcounter{equation}{0}
%--------------------------------------------------------------------------------------

In this section, we give the definition of the model of anisotropic quantum walks that
we consider, we state our main results on quantum walks, and we present the main tools
we use for the proofs. These tools are results of independent interest on commutator
methods for unitary operators in a two-Hilbert spaces setting. The proofs of our
results on commutator methods are given in Section \ref{Sec_unitary} and the proofs of
our results on quantum walks are given in Section \ref{Sec_quantum}.

Let us consider the Hilbert space of square-summable $\C^2$-valued sequences
$$
\H:=\ell^2(\Z,\C^2)
=\left\{\Psi:\Z\to\C^2\mid\sum_{x\in\Z}\|\Psi(x)\|_2^2<\infty\right\},
$$
where $\|\cdot\|_2$ is the usual norm on $\C^2$. The evolution operator of the
one-dimensional quantum walk in $\H$ that we consider is defined by $U:=SC$, with $S$
a shift operator and $C$ a coin operator defined as follows. The operator $S$ is given
by
$$
(S\Psi)(x)
:=\begin{pmatrix}
\Psi^{(0)}(x+1)\\
\Psi^{(1)}(x-1)
\end{pmatrix},
\quad
\Psi
=\begin{pmatrix}
\Psi^{(0)}\\
\Psi^{(1)}
\end{pmatrix}\in\H,~x\in\Z,
$$
and the operator $C$ is given by
$$
(C\Psi)(x):=C(x)\Psi(x),\quad \Psi\in\H,~x\in\Z,~C(x)\in\U(2).
$$
In particular, the evolution operator $U$ is unitary in $\H$ since both $S$ and $C$
are unitary in $\H$.

Throughout the paper, we assume that the coin operator $C$ exhibits an anisotropic
behaviour at infinity. More precisely, we assume that $C$ converges with short-range
rate to two asymptotic coin operators, one on the left and one on the right in the
following way:

\begin{Assumption}[Short-range assumption]\label{ass_short}
There exist $C_\ell,C_{\rm r}\in\U(2)$, $\kappa_\ell,\kappa_{\rm r}>0$, and
$\varepsilon_\ell,\varepsilon_{\rm r}>0$ such that
\begin{align*}
&\big\|C(x)-C_\ell\big\|_{\B(\C^2)}
\le\kappa_\ell\;\!|x|^{-1-\varepsilon_\ell}\quad\hbox{if $x<0$}\\
&\big\|C(x)-C_{\rm r}\big\|_{\B(\C^2)}
\le\kappa_{\rm r}\;\!|x|^{-1-\varepsilon_{\rm r}}\quad\hbox{if $x>0$,}
\end{align*}
where the indexes $\ell$ and ${\rm r}$ stand for ``left" and ``right".
\end{Assumption}

This assumption provides us with two new unitary operators
\begin{equation}\label{eq_def_U_ell}
U_\ell:=SC_\ell\quad\hbox{and}\quad U_{\rm r}:=SC_{\rm r}
\end{equation}
describing the asymptotic behaviour of $U$ on the left and on the right. The precise
sense (from the scattering point of view) in which the operators $U_\ell$ and
$U_{\rm r}$ describe the asymptotic behaviour of $U$ on the left and on the right will
be given in \cite{RST_2}, and the spectral properties of $U_\ell$ and $U_{\rm r}$ are
determined in Section \ref{Sec_asymptotic}. Here, we just introduce the set
$$
\tau(U):=\partial\sigma(U_\ell)\cup\partial\sigma(U_{\rm r}),
$$
where $\partial\sigma(U_\ell),\partial\sigma(U_{\rm r})$ denote the boundaries in the
unit circle $\T:=\{z\in\C\mid|z|=1\}$ of the spectra
$\sigma(U_\ell),\sigma(U_{\rm r})$ of $U_\ell,U_{\rm r}$. In Section
\ref{Sec_asymptotic}, we show that $\tau(U)$ is finite and can be interpreted as the
set of thresholds in the spectrum of $U$.

Our main results on the operator $U$, proved in Sections \ref{Sec_Mourre} and
\ref{Sec_spectrum}, are the  following three theorems on locally $U$-smooth operators
and on the structure of the spectrum of $U$. The symbols $\sigma_{\rm ess}(U)$,
$\sigma_{\rm p}(U)$ and $Q$ stand for the essential spectrum of $U$, the pure point
spectrum of $U$, and the position operator in $\H$, respectively (see \eqref{eq_def_Q}
for precise definition of $Q$).

\begin{Theorem}[Essential spectrum of $U$]\label{thm_essential}
One has $\sigma_{\rm ess}(U)=\sigma(U_\ell)\cup\sigma(U_{\rm r})$.
\end{Theorem}

\begin{Theorem}[$U$-smooth operators]\label{thm_U_smooth_walks}
Let $\G$ be an auxiliary Hilbert space and let $\Theta\subset\T$ be an open set with
closure $\overline\Theta\subset\T\setminus\tau(U)$. Then each operator $T\in\B(\H,\G)$
which extends continuously to an element of
$\B\big(\dom(\langle Q\rangle^{-s}),\G\big)$ for some $s>1/2$ is locally $U$-smooth on
$\Theta\setminus\sigma_{\rm p}(U)$.
\end{Theorem}

\begin{Theorem}[Spectrum of $U$]\label{thm_spectrum_U_walks}
For any closed set $\Theta\subset\T\setminus\tau(U)$, the operator $U$ has at most
finitely many eigenvalues in $\Theta$, each one of finite multiplicity, and $U$ has no
singular continuous spectrum in $\Theta$.
\end{Theorem}

To prove these theorems, we develop in Section \ref{Sec_unitary} commutator methods
for unitary operators in a two-Hilbert spaces setting: Given a triple $(\H,U,A)$
consisting in a Hilbert space $\H$, a unitary operator $U$, and a self-adjoint
operator $A$, we determine how to obtain commutator results for $(\H,U,A)$ in terms of
commutator results for a second triple $(\H_0,U_0,A_0)$ also consisting in a Hilbert
space, a unitary operator, and a self-adjoint operator. In the process, a bounded
identification operator $J:\H_0\to\H$ must also be chosen. The intuition behind this
approach comes from scattering theory which tells us that given a unitary operator $U$
describing some quantum system in a Hilbert space $\H$ there often exists a simpler
unitary operator $U_0$ in a second Hilbert space $\H_0$ describing the same quantum
system in some asymptotic regime.

Our main results in this context are the following. First, we present in Theorem
\ref{thm_fonctionrho} conditions guaranteeing that $U$ and $A$ satisfy a Mourre
estimate on a Borel set $\Theta\subset\T$ as soon as $U_0$ and $A_0$ satisfy a Mourre
estimate on $\Theta$ (equivalently, we present conditions guaranteeing that $A$ is a
conjugate operator for $U$ on $\Theta$ as soon as $A_0$ is a conjugate operator for
$U_0$ on $\Theta$). Next, we present in Proposition \ref{prop_C1_short} conditions
guaranteeing that $U$ is regular with respect to $A$ (that is, $U\in C^1(A)$) as soon
as $U_0$ is regular with respect to $A_0$ (that is, $U_0\in C^1(A_0)$). Finally, we
give in Assumption \ref{ass_eaa} and Corollaries
\ref{Corol_C1(A)}-\ref{Corol_est_supp} conditions guaranteeing that the most natural
choice for the operator $A$, namely $A=JA_0J^*$, is indeed a conjugate operator for
$U$ as soon as $A_0$ is a conjugate operator for $U_0$.

%--------------------------------------------------------------------------------------
\section{Unitary operators in a two-Hilbert spaces setting}\label{Sec_unitary}
\setcounter{equation}{0}
%--------------------------------------------------------------------------------------

In this section, we recall some facts on the spectral family of unitary operators, the
Cayley transform of a unitary operator, locally smooth operators for unitary
operators, and commutator methods for unitary operators in one Hilbert space. We also
present new results on commutator methods for unitary operators in a two-Hilbert
spaces setting.

%--------------------------------------------------------------------------------------
\subsection{Cayley transform}\label{Sec_Cayley}
%--------------------------------------------------------------------------------------

Let $\H$ be a Hilbert space with norm $\|\cdot\|_\H$ and scalar product
$\langle\;\!\cdot\;\!,\;\!\cdot\;\!\rangle_\H$ linear in the second argument, let
$\B(\H)$ be the set of bounded linear operators in $\H$ with norm
$\|\cdot\|_{\B(\H)}$, and let $\K(\H)$ be the set of compact linear operators in $\H$.
A unitary operator $U$ in $\H$ is a surjective isometry, that is, an element
$U\in\B(\H)$ satisfying $U^*U=UU^*=1$. Since $U^*U=UU^*$, the spectral theorem for
normal operators implies that $U$ admits exactly one complex spectral family $E_U$,
with support $\supp(E_U)\subset\T$, such that $U=\int_\C z\;\!E_U(\d z)$. The support
$\supp(E_U)$ is the set of points of non-constancy of $E_U$, which coincides with the
spectrum $\sigma(U)$ of $U$ \cite[Thm.~7.34(a)]{Wei80}. For each $s,t\in\R$, one has
the factorization
$$
E_U(s+it):=E_{\re(U)}(s)\;\!E_{\im(U)}(t),
$$
where $E_{\re(U)}$ and $E_{\im(U)}$ are the real spectral families of the bounded
self-adjoint operators
$$
\textstyle
\re(U):=\frac12\;\!(U+U^*)\quad\hbox{and}\quad
\im(U):=\frac1{2i}\;\!(U-U^*).
$$
One can associate in a canonical way a real spectral family $\widetilde E_U$, with
support $\supp(\widetilde E_U)\subset[0,2\pi]$, to the complex spectral family $E_U$
by noting that
$$
U=\int_\R\e^{i\lambda}\widetilde E_U(\d \lambda)
\quad\hbox{with}\quad
\widetilde E_U(\lambda):=
\begin{cases}
0 & \hbox{if }\lambda<0\\
E_U(\e^{i\lambda}) & \hbox{if }\lambda\in[0,2\pi)\\
1 & \hbox{if }\lambda\ge2\pi.
\end{cases}
$$
Since $\widetilde E_U$ is a real spectral family, the corresponding real spectral
measure $\widetilde E^U$ admits the decomposition
$$
\widetilde E^U=\widetilde E^U_{\rm p}+\widetilde E^U_{\rm sc}+\widetilde E^U_{\rm ac},
$$
with $\widetilde E^U_{\rm p}$, $\widetilde E^U_{\rm sc}$, $\widetilde E^U_{\rm ac}$
the pure point, the singular continuous, and the absolutely continuous components of
$\widetilde E_U$, respectively. The corresponding subspaces
$\H_{\rm p}(U):=\widetilde E^U_{\rm p}(\R)\H$,
$\H_{\rm sc}(U):=\widetilde E^U_{\rm sc}(\R)\H$,
$\H_{\rm ac}(U):=\widetilde E^U_{\rm ac}(\R)\H$ provide an orthogonal decomposition
$$
\H=\H_{\rm p}(U)\oplus\H_{\rm sc}(U)\oplus\H_{\rm ac}(U)
$$
which reduces the operator $U$. The sets
\begin{equation*}
\sigma_{\rm p}(U):=\sigma\big(U|_{\H_{\rm p}(U)}\big),\quad
\sigma_{\rm sc}(U):=\sigma\big(U|_{\H_{\rm sc}(U)}\big),\quad
\sigma_{\rm ac}(U):=\sigma\big(U|_{\H_{\rm ac}(U)}\big),
\end{equation*}
are called pure point spectrum, singular continuous spectrum, and absolutely spectrum
continuous of $U$, respectively, and the set
$\sigma_{\rm c}(U):=\sigma_{\rm sc}(U)\cup\sigma_{\rm ac}(U)$ is called the continuous
spectrum of $U$.

If $1\notin\sigma_{\rm p}(U)$, then the subspace $(1-U)\;\!\H$ is dense in $\H$, and
the Cayley transform of $U$ given by
\begin{equation}\label{Cayley_U}
H\varphi:=i(1+U)(1-U)^{-1}\varphi,\qquad\varphi\in\dom(H):=(1-U)\;\!\H,
\end{equation}
is a self-adjoint operator in $\H$ \cite[Thm.~8.4(b)]{Wei80}. Also, a simple
calculation shows that
\begin{equation}\label{Cayley_L}
U=(H-i)(H+i)^{-1}=\e^{iL}\quad\hbox{with}\quad L:=2\arctan(H)+\pi.
\end{equation}
Therefore, the points of the spectra $\sigma(L)\subset[0,2\pi]$ of $L$ and
$\sigma(U)\subset\T$ of $U$ are linked by the relation
$$
\theta\in\sigma(U)~\Leftrightarrow~
2\arctan\left(i\;\!\frac{1+\theta}{1-\theta}\right)+\pi\in\sigma(L)
$$
(in particular, the point $\theta=1$ in $\sigma(U)$ corresponds to the points
$\lambda=0$ and $\lambda=2\pi$ in $\sigma(L)$). In consequence, if $E^L$ denotes the
real spectral measure of $L$, one has for any Borel set $\Theta\subset\T$ the equality
\begin{equation}\label{eq_measures}
E^U(\Theta)=E^L(\Lambda)
\quad\hbox{with}\quad
\Lambda:=\left\{2\arctan\left(i\;\!\frac{1+\theta}{1-\theta}\right)+\pi
\mid\theta\in\Theta\right\}.
\end{equation}
This implies for each Borel set $\Lambda\subset[0,2\pi)$ that
$$
\widetilde E^U(\Lambda)=E^U(\e^{i\Lambda})=E^L\big(f(\Lambda)\big)
\quad\hbox{with}\quad
f(\lambda):=
\begin{cases}
0 & \hbox{if}~~\lambda=0\\
2\arctan\left(i\;\!\frac{1+\e^{i\lambda}}{1-\e^{i\lambda}}\right)+\pi
& \hbox{if}~~\lambda\in(0,2\pi).
\end{cases}
$$
But a simple calculation shows that $f(\lambda)=\lambda$ for each
$\lambda\in[0,2\pi)$. So, one has $\widetilde E^U(\Lambda)=E^L(\Lambda)$ for each
Borel set $\Lambda\subset[0,2\pi)$. Now, it is also clear from the definitions that
$\widetilde E^U(\Lambda)=E^L(\Lambda)$ for each Borel set
$\Lambda\subset\R\setminus[0,2\pi)$. So, one concludes that $\widetilde E^U=E^L$, and
thus that $U$ and $L$ possess the same spectral properties, up to the correspondence
$U=\e^{iL}$.

%--------------------------------------------------------------------------------------
\subsection{Locally $U$-smooth operators}\label{Sec_smooth}
%--------------------------------------------------------------------------------------

Let $U$ be a unitary operator in a Hilbert space $\H$, and let $\G$ be an auxiliary
Hilbert space. Then, an operator $T\in\B(\H,\G)$ is locally $U$-smooth on an open set
$\Theta\subset\T$ if for each closed set $\Theta'\subset\Theta$ there exists
$c_{\Theta'}\ge0$ such that
\begin{equation}\label{def_U_smooth}
\sum_{n\in\Z}\big\|T\;\!U^nE^U(\Theta')\varphi\big\|_\G^2
\le c_{\Theta'}\;\!\|\varphi\|_\H^2\quad\hbox{for each $\varphi\in\H$},
\end{equation}
and $T$ is (globally) $U$-smooth if \eqref{def_U_smooth} is satisfied with
$\Theta'=\T$. The condition \eqref{def_U_smooth} is invariant under rotation by
$\omega\in\T$ in the sense that if $T$ is $U$-smooth on $\Theta$, then $T$ is
$(\omega U)$-smooth on $\omega\Theta$ since
$$
\big\|T(\omega U)^nE^{\omega U}(\omega\Theta')\varphi\big\|_\G
=\big\|T\;\!U^nE^U(\Theta')\varphi\big\|_\G
$$
for each closed set $\Theta'\subset\Theta$ and each $\varphi\in\H$. An important
consequence of the existence of a locally $U$-smooth operator $T$ on $\Theta$ is the
inclusion $\overline{E^U(\Theta)T^*\G^*}\subset\H_{\rm ac}(U)$, with $\G^*$ the
adjoint space of $\G$ (see \cite[Thm.~2.1]{ABCF06} for a proof).

Local smoothness with respect to a self-adjoint operator $H$ in $\H$ with domain
$\dom(H)$ is defined in a similar way. An operator $T\in\B\big(\dom(H),\G\big)$ is
locally $H$-smooth on an open set $\Lambda\subset\R$ if for each compact set
$\Lambda'\subset\Lambda$ there exists $c_{\Lambda'}\ge0$ such that
\begin{equation}\label{def_T_smooth}
\int_\R\big\|T\e^{-itH}E^H(\Lambda')\varphi\big\|^2_\G\,\d t
\le c_{\Lambda'}\;\!\|\varphi\|_\H^2\quad\hbox{for each $\varphi\in\H$,}
\end{equation}
and $T$ is (globally) $H$-smooth if \eqref{def_T_smooth} is satisfied with
$\Lambda'=\R$. The condition \eqref{def_T_smooth} is invariant under translation by
$s\in\R$ in the sense that if $T$ is $H$-smooth on $\Lambda$, then $T$ is
$(H+s)$-smooth on $\Lambda+s$ since
$$
\big\|T\e^{-it(H+s)}E^{H+s}(\Lambda'+s)\varphi\big\|_\G
=\big\|T\e^{-itH}E^H(\Lambda')\varphi\big\|_\G
$$
for each compact set $\Lambda'\subset\Lambda$ and each $\varphi\in\H$. Also, the
existence of a locally $H$-smooth operator $T$ on $\Lambda\subset\R$ implies the
inclusion $\overline{E^H(\Lambda)T^*\G^*}\subset\H_{\rm ac}(H)$ (see
\cite[Cor~7.1.2]{ABG96} for a proof).

If $1\notin\sigma_{\rm p}(U)$, then the Cayley transform $H$ of $U$ and the operator
$L=2\arctan(H)+\pi$ are defined by \eqref{Cayley_U} and \eqref{Cayley_L}, and the
existence of locally $U$-smooth operators is equivalent to the existence of locally
$H$-smooth operators and locally $L$-smooth operators:

\begin{Lemma}
Let $U$ be a unitary operator in a Hilbert space $\H$ with $1\notin\sigma_{\rm p}(U)$,
let $\G$ be an auxiliary Hilbert space, let $T\in\B(\H,\G)$, and let $\Theta\subset\T$
be an open set. Then, the following are equivalent:
\begin{enumerate}
\item[(i)] $T$ is locally $U$-smooth on $\Theta$,
\item[(ii)] $T$ is locally $L$-smooth on
$\left\{2\arctan\big(i\;\!\frac{1+\theta}{1-\theta}\big)
+\pi\mid\theta\in\Theta\right\}$,
\item[(iii)] $T(H+i)$ is locally $H$-smooth on
$\left\{i\;\!\frac{1+\theta}{1-\theta}\mid\theta\in\Theta\right\}$.
\end{enumerate}
\end{Lemma}

The equivalence (i) $\Leftrightarrow$ (ii) in the case $\Theta=\T$ is due to T. Kato
(see \cite[Sec.~7]{Kat68}).

\begin{proof}
Assume that $T$ is locally $U$-smooth on $\Theta$, take a closed set
$\Theta'\subset\Theta$, and let
$$
\Lambda':=\left\{2\arctan\left(i\;\!\frac{1+\theta}{1-\theta}\right)+\pi
\mid\theta\in\Theta'\right\}.
$$
Then, Equations \eqref{Cayley_L}-\eqref{eq_measures} and Tonnelli's theorem imply for
each $\varphi\in\H$ that
\begin{align*}
\int_\R \big\|T\e^{-itL}E^L(\Lambda')\varphi\big\|^2_\G\,\d t
&=\sum_{n\in\Z}\int_0^1\big\|T\e^{-i(s+n)L}E^L(\Lambda')\varphi\big\|^2_\G\,\d s\\
&=\int_0^1 \sum_{n\in\Z}\big\|TU^{-n}E^U(\Theta')\e^{-isL}\varphi\big\|^2_\G\,\d s\\
&\le\int_0^1 c_{\Theta'}\;\!\big\|\e^{-isL}\varphi\big\|^2_\H\,\d s\\
&=c_{\Theta'}\;\!\|\varphi\|^2_\H\;\!.
\end{align*}
This shows the implication (i) $\Rightarrow$ (ii). The implication (ii) $\Rightarrow$
(i) is shown in a similar way.

To show the equivalence (i) $\Leftrightarrow$ (iii) we observe that
\eqref{def_U_smooth} is equivalent to
$$
\sup_{z\in\D,\,\psi\in\G,\,\|\psi\|_\G=1}\;\!
\left|\big\langle\psi,T\;\!\delta(U,z)E^U(\Theta')T^*\psi\big\rangle_\G\right|
<\infty,
$$
with $\D:=\{z\in\C\mid|z|<1\}$ and
$
\delta(U,z):=(1-zU^{-1})^{-1}-(1-\overline z^{-1}U^{-1})^{-1}
$
(this follows from the proof of \cite[Thm.~2.2]{ABCF06}), and we observe that
\eqref{def_T_smooth} is equivalent to
$$
\sup_{\omega\in\HH,\,\psi\in\G,\,\|\psi\|_\G=1}\;\!\left|\big\langle\psi,
T\im\big((H-\omega)^{-1}\big)E^H(\Lambda')T^*\psi\big\rangle_\G\right|
<\infty
$$
with $\HH:=\{z\in\C\mid\im(z)>0\}$ (this follows from \cite[Prop.~7.1.1]{ABG96}).
Also, we note that
$$
\delta(U,z)=\big(H^2+1\big)\im\left(\left(H-i\;\!\frac{1+z}{1-z}\right)^{-1}\right),
\quad z\in\D,
$$
and we recall that the map $\D\ni z\mapsto i\;\!\frac{1+z}{1-z}\in\HH$ (the Cayley
transform) is a bijection. So, for any closed set $\Theta'\subset\Theta$, we have
\begin{align*}
&\sup_{z\in\D,\,\psi\in\G,\,\|\psi\|_\G=1}\;\!
\left|\big\langle\psi,T\;\!\delta(U,z)E^U(\Theta')T^*\psi\big\rangle_\G\right|\\
&=\sup_{\omega\in\HH,\,\psi\in\G,\,\|\psi\|_\G=1}\;\!
\left|\big\langle\psi,T(H+i)\im\big((H-\omega)^{-1}\big)E^H(\Lambda')
\big(T(H+i)\big)^*\psi\big\rangle_\G\right|
\end{align*}
with $\Lambda'=\big\{i\;\!\frac{1+\theta}{1-\theta}\mid\theta\in\Theta'\big\}$, and
thus (i) and (iii) are equivalent (note that the operator
$$
T(H+i)\im\big((H-\omega)^{-1}\big)E^H(\Lambda')\big(T(H+i)\big)^*
$$
belongs to $\B(\G)$ for each $\omega\in\HH$ even if $1\in\Theta'$, that is, even if
$\Lambda'$ is not bounded).
\end{proof}

%--------------------------------------------------------------------------------------
\subsection{Commutator methods in one Hilbert space}\label{Sec_one_Hilbert}
%--------------------------------------------------------------------------------------

In this section, we present some results on commutator methods for unitary operators
in one Hilbert space $\H$. We start by recalling definitions and results borrowed from
\cite{ABG96,FRT13,Sah97_2}. Let $S\in\B(\H)$ and let $A$ be a self-adjoint operator in
$\H$ with domain $\dom(A)$. For any $k\in\N$, we say that $S$ belongs to $C^k(A)$,
with notation $S\in C^k(A)$, if the map
$$
\R\ni t\mapsto\e^{-itA}S\e^{itA}\in\B(\H)
$$
is strongly of class $C^k$. In the case $k=1$, one has $S\in C^1(A)$ if and only if
the quadratic form
$$
\dom(A)\ni\varphi\mapsto\big\langle A\;\!\varphi,S\varphi\big\rangle_\H
-\big\langle\varphi,SA\;\!\varphi\big\rangle_\H\in\C
$$
is continuous for the topology induced by $\H$ on $\dom(A)$. The operator
corresponding to the continuous extension of the form is denoted by $[A,S]\in\B(\H)$,
and it verifies
$$
[A,S]=\slim_{\tau\to0}\;\![A_\tau,S]
\quad\hbox{with}\quad A_\tau:=(i\tau)^{-1}\big(\e^{i\tau A}-1\big)\in\B(\H),
\quad\tau\in\R\setminus\{0\}.
$$

Three regularity conditions slightly stronger than $S\in C^1(A)$ are defined as
follows: $S$ belongs to $C^{1,1}(A)$, with notation $S\in C^{1,1}(A)$, if
$$
\int_0^1\big\|\e^{-itA}S\e^{itA}+\e^{itA}S\e^{-itA}-2S\big\|_{\B(\H)}
\,\frac{\d t}{t^2}<\infty.
$$
$S$ belongs to $C^{1+0}(A)$, with notation $S\in C^{1+0}(A)$, if $S\in C^1(A)$ and
$$
\int_0^1\big\|\e^{-itA}[A,S]\e^{itA}-[A,S]\big\|_{\B(\H)}\,\frac{\d t}t<\infty.
$$
$S$ belongs to $C^{1+\varepsilon}(A)$ for some $\varepsilon\in(0,1)$, with notation
$S\in C^{1+\varepsilon}(A)$, if $S\in C^1(A)$ and
$$
\big\|\e^{-itA}[A,S]\e^{itA}-[A,S]\big\|_{\B(\H)}
\le{\rm Const.}\;\!t^\varepsilon\quad\hbox{for all $t\in(0,1)$.}
$$
As banachisable topological vector spaces, the sets $C^2(A)$, $C^{1+\varepsilon}(A)$,
$C^{1+0}(A)$, $C^{1,1}(A)$, $C^1(A)$, and $C^0(A)=\B(\H)$, satisfy the continuous
inclusions \cite[Sec.~5.2.4]{ABG96}
$$
C^2(A)\subset C^{1+\varepsilon}(A)\subset C^{1+0}(A)\subset C^{1,1}(A)\subset C^1(A)
\subset C^0(A).
$$

Now, we adapt to the case of unitary operators the definition of two useful functions
introduced in \cite[Sec.~7.2]{ABG96} in the case of self-adjoint operators. For that
purpose, we let $U$ be a unitary operator with $U\in C^1(A)$, for $S,T\in\B(\H)$ we
write $T\gtrsim S$ if there exists an operator $K\in\K(\H)$ such that $T+K\ge S$, and
for $\theta\in\T$ and $\varepsilon>0$ we set
$$
\Theta(\theta;\varepsilon)
:=\big\{\theta'\in\T\mid|\arg(\theta-\theta')|<\varepsilon\big\}
\quad\hbox{and}\quad
E^U(\theta;\varepsilon):=E^U\big(\Theta(\theta;\varepsilon)\big).
$$
With these notations at hand, we define the functions
$\varrho^A_U:\T\to(-\infty,\infty]$ and $\widetilde\varrho^A_U:\T\to(-\infty,\infty]$
by
$$
\varrho^A_U(\theta)
:=\sup\big\{a\in\R\mid\exists\;\!\varepsilon>0~\hbox{such that}~E^U(\theta;\varepsilon)
U^{-1}[A,U]E^U(\theta;\varepsilon)\ge a\;\!E^U(\theta;\varepsilon)\big\}
$$
and
$$
\widetilde\varrho^A_U(\theta)
:=\sup\big\{a\in\R \mid\exists\;\!\varepsilon>0~\hbox{such that}~E^U(\theta;\varepsilon)
U^{-1}[A,U]E^U(\theta;\varepsilon)\gtrsim a\;\!E^U(\theta;\varepsilon)\big\}.
$$
In applications, the function $\widetilde\varrho^A_U$ is more convenient than the
function $\varrho^A_U $ since it is defined in terms of a weaker positivity condition
(positivity up to compact terms). A simple argument shows that
$\widetilde\varrho^A_U(\theta)$ can be defined in an equivalent way by
\begin{equation}\label{equivalent_def}
\widetilde\varrho^A_U(\theta)
=\sup\big\{a\in\R\mid\exists\;\!\eta\in C^\infty(\T,\R)~\hbox{such that}
~\eta(\theta)\ne0~\hbox{and}~\eta(U)U^{-1}[A,U]\eta(U)\gtrsim a\;\!\eta(U)^2\big\}.
\end{equation}
Further properties of the functions $\widetilde\varrho^A_U$ and $\varrho^A_U$ are
collected in the following lemmas. The first one corresponds to
\cite[Prop.~2.3]{FRT13}.

\begin{Lemma}[Virial Theorem for $U$]
Let $U$ be a unitary operator in $\H$ and let $A$ be a self-adjoint operator in $\H$
with $U\in C^1(A)$. Then,
$$
E^U(\{\theta\})U^{-1}[A,U]E^U(\{\theta\})=0
$$
for each $\theta\in\T$. In particular, one has
$\big\langle\varphi,U^{-1}[A,U]\varphi\big\rangle=0$ for each eigenvector
$\varphi\in\H$ of $U$.
\end{Lemma}

\begin{Lemma}\label{lemma_inequalities}
Let $U$ be a unitary operator in $\H$ and let $A$ be a self-adjoint operator in $\H$
with $U\in C^1(A)$. Assume there exist an open set $\Theta\subset\T$ and $a\in\R$ such
that $E^U(\Theta)U^{-1}[A,U]E^U(\Theta)\gtrsim a\;\!E^U(\Theta)$. Then, for each
$\theta\in\Theta$ and each $\eta>0$ there exist $\varepsilon>0$ and a finite rank
orthogonal projection $F$ with $E^U(\{\theta\})\ge F$ such that
$$
E^U(\theta;\varepsilon)U^{-1}[A,U]E^U(\theta;\varepsilon)
\ge(a-\eta)\big(E^U(\theta;\varepsilon)-F\big)-\eta F.
$$
In particular, if $\theta$ is not an eigenvalue of $U$, then
$$
E^U(\theta;\varepsilon)U^{-1}[A,U]E^U(\theta;\varepsilon)
\ge(a-\eta)E^U(\theta;\varepsilon),
$$
while if $\theta$ is an eigenvalue of $U$, one has only
$$
E^U(\theta;\varepsilon)U^{-1}[A,U]E^U(\theta;\varepsilon)
\ge\min\{a-\eta,-\eta\}E^U(\theta;\varepsilon).
$$
\end{Lemma}

\begin{proof}
The proof uses Virial Theorem for $U$ and is analogous to the proof of
\cite[Lemma~7.2.12]{ABG96} in the self-adjoint case. One just needs to replace in the
proof of \cite[Lemma~7.2.12]{ABG96} $[iH,A]$ by $U^{-1}[A,U]$, $E(J)$ by
$E^U(\Theta)$, $E(\{\lambda\})$ by $E^U(\{\theta\})$, and $E(\lambda;1/k)$ by
$E^U(\theta;1/k)$.
\end{proof}

\begin{Lemma}\label{lemma_properties}
Let $U$ be a unitary operator in $\H$ and let $A$ be a self-adjoint operator in $\H$
with $U\in C^1(A)$.
\begin{enumerate}
\item[(a)] The function $\varrho^A_U:\T\to(-\infty,\infty]$ is lower semicontinuous,
and $\varrho^A_U(\theta)<\infty$ if and only if $\theta\in\sigma(U)$.
\item[(b)] The function $\widetilde\varrho^A_U:\T\to(-\infty,\infty]$ is lower
semicontinuous, and $\widetilde\varrho^A_U(\theta)<\infty$ if and only if
$\theta\in\sigma_{\rm ess}(U)$.
\item[(c)] $\widetilde\varrho^A_U\ge\varrho^A_U$.
\item[(d)] If $\theta\in\T$ is an eigenvalue of $U$ and
$\widetilde\varrho^A_U(\theta)>0$, then $\varrho^A_U(\theta)=0$. Otherwise,
$\varrho^A_U(\theta)=\widetilde\varrho^A_U(\theta)$.
\end{enumerate}
\end{Lemma}

\begin{proof}
The proof is an adaptation of the proofs of Lemma 7.2.1, Proposition 7.2.3(a),
Proposition 7.2.6 and Theorem 7.2.13 of \cite{ABG96} to the case of unitary operators.

(a) The fact that $\varrho^A_U(\theta)<\infty$ if and only if $\theta\in\sigma(U)$
follows from the definition of $\varrho^A_U$ and the closedness of $\sigma(U)$. Let
$\theta_0\in\T$ and let $r\in\R$ be such that $\varrho^A_U(\theta_0)>r$. To show the
lower semicontinuity of $\varrho^A_U$ we must show that there is a neighbourhood of
$\theta_0$ on which $\varrho^A_U>r$. Since $\varrho^A_U(\theta_0)>r$, there exist
$a>r$ and $\varepsilon_0>0$ such that
$$
E^U(\theta_0;\varepsilon_0)U^{-1}[A,U]E^U(\theta_0;\varepsilon_0)
\ge a\;\!E^U(\theta_0;\varepsilon_0).
$$
Let $\varepsilon:=\varepsilon_0/2$ and $\theta\in\Theta(\theta_0;\varepsilon_0)$. By
multiplying on the left and on the right the preceding inequality by
$E^U(\theta;\varepsilon)$ and by using the fact that
$
E^U(\theta;\varepsilon)E^U(\theta_0;\varepsilon_0)=E^U(\theta;\varepsilon)
$,
one obtains
$$
E^U(\theta;\varepsilon)U^{-1}[A,U]E^U(\theta;\varepsilon)
\ge a\;\!E^U(\theta;\varepsilon).
$$
This implies that $\varrho^A_U(\theta)\ge a>r$ for all
$\theta\in\Theta(\theta_0;\varepsilon_0)$.

(b)-(c) The lower semicontinuity of $\widetilde\varrho^A_U$ is obtained similarly to
that of $\varrho^A_U$ in point (a), and the inequality
$\widetilde\varrho^A_U\ge\varrho^A_U$ is immediate from the definitions. For the last
claim, we use the fact that $\theta\notin\sigma_{\rm ess}(U)$ if and only if
$E^U(\theta;\varepsilon)\in\K(\H)$ for some $\varepsilon>0$. So
$\theta\notin\sigma_{\rm ess}(U)$ implies that $\widetilde\varrho^A_U(\theta)=\infty$.
Conversely, if $\widetilde\varrho^A_U(\theta)=\infty$, let
$m:=\|E^U(\theta;1)U^{-1}[A,U]E^U(\theta;1)\|_{\B(\H)}$ and $a>m$. Then, there is
$\varepsilon\in(0,1)$ such that
$$
E^U(\theta;\varepsilon)U^{-1}[A,U]E^U(\theta;\varepsilon)
\gtrsim a\;\!E^U(\theta;\varepsilon).
$$
On another hand, the inequality $m\ge E^U(\theta;1)U^{-1}[A,U]E^U(\theta;1)$ and the
fact that $E^U(\theta;\varepsilon)E^U(\theta;1)=E^U(\theta;\varepsilon)$ imply that
$$
m\;\!E^U(\theta;\varepsilon)\ge
E^U(\theta;\varepsilon)U^{-1}[A,U]E^U(\theta;\varepsilon).
$$
Thus $m\;\!E^U(\theta;\varepsilon)\gtrsim a\;\!E^U(\theta;\varepsilon)$, and there
exists $K\in\K(\H)$ such that $K\ge(a-m)\;\!E^U(\theta;\varepsilon)$. This implies by
heredity of compactness that $E^U(\theta;\varepsilon)\in\K(\H)$.

(d) If $\theta$ is not an eigenvalue of $U$, then Lemma \ref{lemma_inequalities}
implies that $\widetilde\varrho^A_U(\theta)\le\varrho^A_U(\theta)$, and so these two
numbers must be equal by point (c). Now assume that $\theta$ is an eigenvalue of $U$.
If $\widetilde\varrho^A_U(\theta)\le0$, then $a\le0$ in Lemma
\ref{lemma_inequalities}, hence $\min\{a-\eta,-\eta\}=a-\eta$ and we have the same
result as before. If $\widetilde\varrho^A_U(\theta)>0$, we may take $a>0$ in Lemma
\ref{lemma_inequalities}, which leads to the inequality $\varrho^A_U(\theta)\ge0$; the
opposite inequality $\varrho^A_U(\theta)\le0$ follows by using Virial theorem for $U:$
if $a<\varrho^A_U(\theta)$, there is $\varepsilon>0$ such that
$E^U(\theta;\varepsilon)U^{-1}[A,U]E^U(\theta;\varepsilon)
\ge a\;\!E^U(\theta;\varepsilon)$; hence
$0=E^U(\{\theta\})U^{-1}[A,U]E^U(\{\theta\})\ge a\;\!E^U(\{\theta\})$. Since
$E^U(\{\theta\})\ne0$, we must have $a\le0$.
\end{proof}

By analogy with the self-adjoint case, we say that $A$ is conjugate to $U$ at the
point $\theta\in\T$ if $\widetilde\varrho^A_U(\theta)>0$, and that $A$ is strictly
conjugate to $U$ at $\theta$ if $\varrho^A_U(\theta)>0$. Since
$\widetilde\varrho^A_U(\theta)\ge\varrho^A_U(\theta)$ for each $\theta\in\T$ by Lemma
\ref{lemma_properties}(c), strict conjugation is a property stronger than conjugation.

\begin{Theorem}[$U$-smooth operators]\label{thm_U_smooth}
Let $U$ be a unitary operator in $\H$, let $A$ be a self-adjoint operator in $\H$, and
let $\G$ be an auxiliary Hilbert space. Assume either that $U$ has a spectral gap and
$U\in C^{1,1}(A)$, or that $U\in C^{1+0}(A)$. Suppose also there exist an open set
$\Theta\subset\T$, a number $a>0$ and an operator $K\in\K(\H)$ such that
$$
E^U(\Theta)\;\!U^{-1}[A,U]\;\!E^U(\Theta)\ge aE^U(\Theta)+K.
$$
Then, each operator $T\in\B(\H,\G)$ which extends continuously to an element of
$\B\big(\dom(\langle A\rangle^s)^*,\G\big)$ for some $s>1/2$ is locally $U$-smooth on
$\Theta\setminus\sigma_{\rm p}(U)$.
\end{Theorem}

\begin{proof}
The claim follows by adapting the proof of \cite[Prop.~2.9]{FRT13} to locally
$U$-smooth operators $T$ with values in the auxiliary Hilbert space $\G$, taking into
account the results of Section \ref{Sec_smooth}.
\end{proof}

The last theorem of this section corresponds to \cite[Thm.~2.7]{FRT13}:

\begin{Theorem}[Spectrum of $U$]\label{thm_spec_U}
Let $U$ be a unitary operator in $\H$ and let $A$ be a self-adjoint operator in $\H$.
Assume either that $U$ has a spectral gap and $U\in C^{1,1}(A)$, or that
$U\in C^{1+0}(A)$. Suppose also there exist an open set $\Theta\subset\T$, a number
$a>0$ and an operator $K\in\K(\H)$ such that
$$
E^U(\Theta)\;\!U^{-1}[A,U]\;\!E^U(\Theta)\ge aE^U(\Theta)+K.
$$
Then, $U$ has at most finitely many eigenvalues in $\Theta$, each one of finite
multiplicity, and $U$ has no singular continuous spectrum in $\Theta$.
\end{Theorem}

%--------------------------------------------------------------------------------------
\subsection{Commutator methods in a two-Hilbert spaces setting}\label{Sec_two_Hilbert}
%--------------------------------------------------------------------------------------

From now on, in addition to the triple $(\H,U,A)$, we consider a second triple
$(\H_0,U_0,A_0)$ with $\H_0$ a Hilbert space, $U_0$ a unitary operator in $\H_0$, and
$A_0$ a self-adjoint operator in $\H_0$. We also consider an identification operator
$J\in\B(\H_0,\H)$. The existence of two such triples with an identification operator
is quite standard in scattering theory of unitary operators, at least for the pairs
$(\H,U)$ and $(\H_0,U_0)$ (see for instance the books \cite{BW83,Yaf92}). Part of our
goal in this section is to show that the existence of the conjugate operators $A$ and
$A_0$ is also natural, in the same way it is in the self-adjoint case \cite{RT13_2}.

In the one-Hilbert space setting, the unitary operator $U$ is usually a multiplicative
perturbation of the unitary operator $U_0$. In this case, if $U-U_0$ is compact, the
stability of the function $\widetilde\varrho_{U_0}^{A_0}$ under compact perturbations
allows one to infer information on $U$ from similar information on $U_0$ (see
\cite[Cor.~2.10]{FRT13}). In the two-Hilbert spaces setting, we are not aware of any
general result relating the functions $\widetilde\varrho_U^A$ and
$\widetilde\varrho_{U_0}^{A_0}$. The obvious reason for this being the impossibility
to consider $U$ as a direct perturbation of $U_0$ since these operators do not act in
the same Hilbert space. Nonetheless, the next theorem provides a result in that
direction. For two arbitrary Hilbert spaces $\H_1, \H_2$ and two operators
$S,T\in\B(\H_1,\H_2)$, we use the notation $T\approx S$ if $(T-S)\in\K(\H_1,\H_2)$.

\begin{Theorem}\label{thm_fonctionrho}
Let $(\H_0,U_0,A_0)$ and $(\H,U,A)$ be as above, let $J\in\B(\H_0,\H)$, and assume
that
\begin{enumerate}
\item[(i)] $U_0\in C^1(A_0)$ and $U\in C^1(A)$,
\item[(ii)] $JU_0^{-1}[A_0,U_0]J^*-U^{-1}[A,U]\in\K(\H)$,
\item[(iii)] $JU_0-UJ\in\K(\H_0,\H)$,
\item[(iv)] For each $\eta\in C(\C,\R)$, $\eta(U)(JJ^*-1)\eta(U)\in\K(\H)$.
\end{enumerate}
Then, one has $\widetilde\varrho_U^A\ge \widetilde\varrho_{U_0}^{A_0}$.
\end{Theorem}

An induction argument together with a Stone-Weierstrass density argument shows that
(iii) is equivalent to the apparently stronger condition
\begin{enumerate}
\item[(iii')] For each $\eta\in C(\C,\R)$, $J\eta(U_0)-\eta(U)J\in\K(\H_0,\H)$.
\end{enumerate}
Therefore, in the sequel, we will sometimes use the condition (iii') instead of (iii).

\begin{proof}
For each $\eta\in C(\C,\R)$, we have
\begin{equation}\label{eq_Eone}
\eta(U)U^{-1}[A,U]\eta(U)
\approx\eta(U)JU_0^{-1}[A_0,U_0]J^*\eta(U)\\
\approx J\eta(U_0)U_0^{-1}[A_0,U_0]\eta(U_0)J^*
\end{equation}
due to Assumption (i)-(iii). Furthermore, if there exists $a\in\R$ such that
$$
\eta(U_0)U_0^{-1}[A_0,U_0]\eta(U_0)\gtrsim a\;\!\eta(U_0)^2,
$$
then Assumptions (iii)-(iv) imply that
\begin{equation}\label{eq_Etwo}
J\eta(U_0)U_0^{-1}[A_0,U_0]\eta(U_0)J^*
\gtrsim a\;\!J\eta(U_0)^2J^*
\approx a\;\!\eta(U)JJ^*\eta(U)
\approx a\;\!\eta(U)^2.
\end{equation}
Thus, we obtain $\eta(U)U^{-1}[A,U]\eta(U)\gtrsim a\;\!\eta(U)^2$ by combining
\eqref{eq_Eone} and \eqref{eq_Etwo}. This last estimate, together with the definition
\eqref{equivalent_def} of the functions $\widetilde\varrho_{U_0}^{A_0}$ and
$\widetilde\varrho_U^A$, implies the claim.
\end{proof}

The regularity of $U_0$ with respect to $A_0$ is usually easy to check, while the
regularity of $U$ with respect to $A$ is in general difficult to establish. For that
purpose, various perturbative criteria have been developed for self-adjoint operators
in one Hilbert space, and often a distinction is made between so-called short-range
and long-range perturbations. Roughly speaking, the two terms of the formal commutator
$[A,U]=AU-UA$ are treated separately in the short-range case, while the commutator
$[A,U]$ is really computed in the long-range case. In the sequel, we discuss the case
of short-range type perturbations for unitary operators in a two-Hilbert spaces
setting. The results we obtain are analogous to the ones obtained in
\cite[Sec.~3.1]{RT13_2} for self-adjoint operators in a two-Hilbert spaces setting.

We start by showing how the condition $U\in C^1(A)$ and the assumptions (ii)-(iii) of
Theorem \ref{thm_fonctionrho} can be verified for a class of short-range type
perturbations. Our approach is to infer the desired information on $U$ from equivalent
information on $U_0$, which are usually easier to obtain. Accordingly, our results
exhibit some perturbative flavor. The price one has to pay is to impose some
compatibility conditions between $A_0$ and $A$. For brevity, we set
$$
B:=J U_0-U J\in\B(\H_0,\H)\quad\hbox{and}\quad B_*:=J U_0^*-U^*J\in\B(\H_0,\H).
$$

\begin{Proposition}\label{prop_C1_short}
Let $U_0\in C^1(A_0)$, assume that $\DD\subset\H$ is a core for $A$ such that
$J^*\DD\subset\dom(A_0)$, and suppose that
\begin{equation}\label{hyp1}
\overline{BA_0\upharpoonright\dom(A_0)}\in\B(\H_0,\H),
\quad \overline{B_*A_0\upharpoonright\dom(A_0)}\in\B(\H_0,\H)
\quad\hbox{and}\quad\overline{(JA_0J^*-A)\upharpoonright\DD}\in\B(\H).
\end{equation}
Then, $U\in C^1(A)$.
\end{Proposition}

\begin{proof}
For $\varphi\in\DD$, a direct calculation gives
\begin{align*}
&\big\langle A\varphi,U\varphi\big\rangle_\H
-\big\langle\varphi,UA\varphi\big\rangle_\H\\
&=\big\langle A\varphi,U\varphi\big\rangle_\H
-\big\langle\varphi,UA\varphi\big\rangle_\H
-\big\langle\varphi,J\;\![A_0,U_0]J^*\varphi\big\rangle_\H
+\big\langle\varphi,J\;\![A_0,U_0]J^*\varphi\big\rangle_\H\\
&=\big\langle\varphi,BA_0J^*\varphi\big\rangle_\H
-\big\langle B_*A_0J^*\varphi,\varphi\big\rangle_\H
+\big\langle U^*\varphi,(JA_0J^*-A)\varphi\rangle_\H
-\big\langle(JA_0J^*-A)\varphi,U\varphi\big\rangle_\H\\
&\quad+\big\langle\varphi,J\;\![A_0,U_0]J^*\varphi\big\rangle_\H.
\end{align*}
Furthermore, we have
$$
\big|\big\langle\varphi,BA_0J^*\varphi\big\rangle_\H
-\big\langle B_*A_0J^*\varphi,\varphi\big\rangle_\H\big|
\le{\rm Const.}\;\!\|\varphi\|_\H^2
$$
due to the first two conditions in \eqref{hyp1}, and we have
$$
\big|\big\langle U^*\varphi,(JA_0J^*-A)\varphi\big\rangle_\H
-\big\langle (JA_0J^*-A)\varphi,U \varphi\big\rangle_\H\big|
\le{\rm Const.}\;\!\|\varphi\|^2_\H
$$
due to the third condition in \eqref{hyp1}. Finally, since $U_0\in C^1(A_0)$ and
$J\in\B(\H_0,\H)$ we also have
$$
\big|\big\langle\varphi,J\;\![A_0,U_0]J^*\varphi\big\rangle_\H\big|
\le{\rm Const.}\;\!\|\varphi\|^2_\H.
$$
Since $\DD$ is a core for $A$, this implies that $U\in C^1(A)$.
\end{proof}

We now show how the assumption (ii) of Theorem \ref{thm_fonctionrho} is verified for a
short-range type perturbation. Note that the hypotheses of the following proposition
are slightly stronger than the ones of Proposition \ref{prop_C1_short}. Thus, $U$
automatically belongs to $C^1(A)$.

\begin{Proposition}\label{prop_com_compact}
Let $U_0\in C^1(A_0)$, assume that $\DD\subset\H$ is a core for $A$ such that
$J^*\DD\subset\dom(A_0)$, and suppose that
\begin{equation}\label{c123}
\overline{BA_0\upharpoonright\dom(A_0)}\in\B(\H_0,\H),
\quad\overline{B_*A_0\upharpoonright\dom(A_0)}\in\K(\H_0,\H)
\quad\hbox{and}\quad\overline{(JA_0J^*-A)\upharpoonright\DD}\in\K(\H).
\end{equation}
Then, the difference of bounded operators $JU_0^{-1}[A_0,U_0]J^*-U^{-1}[A,U]$ belongs
to $\K(\H)$.
\end{Proposition}

\begin{proof}
The facts that $U_0\in C^1(A_0)$ and $J^*\DD\subset\dom(A_0)$ imply the inclusions
$$
U_0J^*\DD\subset U_0\;\!\dom(A_0)\subset\dom(A_0).
$$
Using this and the last two conditions of \eqref{c123}, we obtain for $\varphi\in\DD$
and $\psi\in U^{-1}\DD$ that
\begin{align*}
&\big\langle\psi,\big(JU_0^{-1}[A_0,U_0]J^*-U^{-1}[A,U]\big)\varphi\big\rangle_\H\\
&=\big\langle\psi,B_*A_0U_0J^*\varphi\big\rangle_\H
+\big\langle B_*A_0J^*U\psi,\varphi\big\rangle_\H
+\big\langle(JA_0J^*-A)U\psi,U\varphi\big\rangle_\H
-\big\langle\psi,(JA_0J^*-A)\varphi\big\rangle_\H\\
&=\big\langle\psi,K_1U_0J^*\varphi\big\rangle_\H
+\big\langle K_1J^*U\psi,\varphi\big\rangle_\H
+\big\langle K_2U\psi,U\varphi\big\rangle_\H
-\big\langle\psi,K_2\varphi\big\rangle_\H
\end{align*}
with $K_1\in\K(\H_0,\H)$ and $K_2\in\K(\H)$. Since $\DD$ and $U^{-1}\DD$ are dense in
$\H$, it follows that the operator $JU_0^{-1}[A_0,U_0]J^*-U^{-1}[A,U]$ belongs to
$\K(\H)$.
\end{proof}

In the rest of the section, we particularize the previous results to the case where
$A=JA_0J^*$. This case deserves a special attention since it represents the most
natural choice of a conjugate operator $A$ for $U$ when a conjugate operator $A_0$ for
$U_0$ is given. However, one needs in this case the following assumption to guarantee
the self-adjointness of the operator $A:$

\begin{Assumption}\label{ass_eaa}
There exists a set $\DD\subset\dom(A_0J^*)\subset \H$ such that
$JA_0J^*\upharpoonright\DD$ is essentially self-adjoint, with corresponding
self-adjoint extension denoted by $A$.
\end{Assumption}

Assumption \ref{ass_eaa} might be difficult to check in general, but in concrete
situations the choice of the set $\DD$ can be quite natural (see for example Lemma
\ref{lemma_A_self} for the case of quantum walks or \cite[Rem.~4.3]{RT13_1} for the
case of manifolds with asymptotically cylindrical ends). The following two corollaries
follow directly from Propositions \ref{prop_C1_short}-\ref{prop_com_compact} in the
case Assumption \ref{ass_eaa} is satisfied.

\begin{Corollary}\label{Corol_C1(A)}
Let $U_0\in C^1(A_0)$, suppose that Assumption \ref{ass_eaa} holds for some set
$\DD\subset\H$, and assume that
$$
\overline{BA_0\upharpoonright\dom(A_0)}\in\B(\H_0,\H)
\quad\hbox{and}\quad\overline{B_*A_0\upharpoonright\dom(A_0)}\in\B(\H_0,\H).
$$
Then, $U$ belongs to $C^1(A)$.
\end{Corollary}

\begin{Corollary}\label{Corol_est_supp}
Let $U_0\in C^1(A_0)$, suppose that Assumption \ref{ass_eaa} holds for some set
$\DD\subset\H$, and assume that
$$
\overline{BA_0\upharpoonright\dom(A_0)}\in\B(\H_0,\H)
\quad\hbox{and}\quad \overline{B_*A_0\upharpoonright\dom(A_0)}\in\K(\H_0,\H).
$$
Then, the difference of bounded operators $JU_0^{-1}[A_0,U_0]J^*-U^{-1}[A,U]$ belongs
to $\K(\H)$.
\end{Corollary}

%--------------------------------------------------------------------------------------
\section{Quantum walks with an anisotropic coin}\label{Sec_quantum}
\setcounter{equation}{0}
%--------------------------------------------------------------------------------------

In this section, we apply the abstract theory of Section \ref{Sec_unitary} to prove
our results on the spectrum of the evolution operator $U$ of the quantum walk with an
anisotropic coin defined in Section \ref{Sec_model}. For this, we first determine in
Section \ref{Sec_asymptotic} the spectral properties and prove a Mourre estimate for
the asymptotic operators $U_\ell$ and $U_{\rm r}$. Then, in Section \ref{Sec_Mourre},
we use the Mourre estimate for $U_\ell$ and $U_{\rm r}$ to derive a Mourre estimate
for $U$. Finally, in Section \ref{Sec_spectrum}, we use the Mourre estimate for $U$ to
prove our results on $U$. We recall that the behaviour of the coin operator $C$ at
infinity is determined by Assumption \ref{ass_short}.

%--------------------------------------------------------------------------------------
\subsection{Asymptotic operators $U_\ell$ and $U_{\rm r}$}\label{Sec_asymptotic}
%--------------------------------------------------------------------------------------

For the study of the asymptotic operators $U_\ell$ and $U_{\rm r}$, we use the symbol
$\star$ to denote either the index $\ell$ or the index ${\rm r}$. Also, we introduce
the subspace $\H_{\rm fin}\subset\H$ of elements with finite support
$$
\H_{\rm fin}
:=\bigcup_{n\in\N}\big\{\Psi\in\H\mid\hbox{$\Psi(x)=0$ if $|x|\ge n$}\big\},
$$
the Hilbert space $\KK:=\ltwo\big([0,2\pi),\frac{\d k}{2\pi},\C^2\big)$, and the
discrete Fourier transform $\F:\H\to\KK$, which is the unitary operator defined as the
unique continuous extension of the operator
$$
(\F\Psi)(k):=\sum_{x\in\Z}\e^{-ikx}\Psi(x),\quad\Psi\in\H_{\rm fin},~k\in[0,2\pi).
$$
A direct computation shows that the operator $U_\star$ is decomposable in the Fourier
representation, namely, for all $f\in\KK$ and almost every $k\in[0,2\pi)$ we have
$$
(\F\;\!U_\star\;\!\F^*f)(k)=\widehat{U_\star}(k)f(k)\quad\hbox{with}\quad
\widehat{U_\star}(k)
:=\begin{pmatrix}
\e^{ik}&0\\
0&\e^{-ik}
\end{pmatrix}
C_\star\in\U(2).
$$
Moreover, since $\widehat{U_\star}(k)\in\U(2)$ the spectral theorem implies that
$\widehat{U_\star}(k)$ can be written as
$$
\widehat{U_\star}(k)=\sum_{j=1}^2\lambda_{\star,j}(k)\;\!\Pi_{\star,j}(k),
$$
with $\lambda_{\star,j}(k)$ the eigenvalues of $\widehat{U_\star}(k)$ and
$\Pi_{\star,j}(k)$ the corresponding orthogonal projections.

The next lemma furnishes some information on the spectrum of $U_\star$. To state it,
we use the following parametrisation for the matrices
$C_\star:$
\begin{equation}\label{eq_param_C}
C_\star=\e^{i\delta_\star/2}
\begin{pmatrix}
a_\star\e^{i(\alpha_\star-\delta_\star/2)}
& b_\star\e^{i(\beta_\star-\delta_\star/2)}\\
-b_\star\e^{-i(\beta_\star-\delta_\star/2)}
& a_\star\e^{-i(\alpha_\star-\delta_\star/2)}
\end{pmatrix}
\end{equation}
with $a_\star,b_\star\in[0,1]$ satisfying $a_\star^2+b_\star^2=1$, and
$\alpha_\star,\beta_\star,\delta_\star\in(-\pi,\pi]$. The determinant $\det(C_\star)$
of $C_\star$ is equal to $\e^{i\delta_\star}$. For brevity, we also set
\begin{align*}
\tau_\star(k)&:=a_\star\cos(k+\alpha_\star-\delta_\star/2),\\
\eta_\star(k)&:=\sqrt{1-\tau_\star(k)^2},\\
\varsigma_\star(k)&:=a_\star\sin(k+\alpha_\star-\delta_\star/2),\\
\theta_\star&:=\arccos(a_\star).
\end{align*}

\begin{Lemma}[Spectrum of $U_\star$]\label{lemma_spectrum_U_star}
\begin{enumerate}
\item[(a)] If $a_\star=0$, then $U_\star$ has pure point spectrum
$$
\sigma(U_\star)
=\sigma_{\rm p}(U_\star)
=\big\{i\e^{i\delta_\star/2},-i\e^{i\delta_\star/2}\big\}
$$
with each point an eigenvalue of $U_\star$ of infinite multiplicity.
\item[(b)] If $a_\star\in(0,1)$, then $\sigma_{\rm p}(U_\star)=\varnothing$ and
$$
\sigma(U_\star)
=\sigma_{\rm c}(U_\star)
=\big\{\e^{i\gamma}\mid
\gamma\in[\delta_\star/2+\theta_\star,\pi+\delta_\star/2-\theta_\star]
\cup[\pi+\delta_\star/2+\theta_\star,2\pi+\delta_\star/2-\theta_\star]\big\}.
$$
\item[(c)] If $a_\star=1$, then $\sigma_{\rm p}(U_\star)=\varnothing$ and
$\sigma(U_\star)=\sigma_{\rm c}(U_\star)=\T$.
\end{enumerate}
\end{Lemma}

\begin{proof}
Using the parametrisation for $C_\star$ given in \eqref{eq_param_C}, one gets
$$
\widehat{U_\star}(k)
=\e^{i\delta_\star/2}
\begin{pmatrix}
a_\star(k) & b_\star(k)\\
-\overline{b_\star(k)} & \overline{a_\star(k)}
\end{pmatrix}
$$	
with
$$
a_\star(k):=a_\star\e^{i(k+\alpha_\star-\delta_\star/2)}
\quad\hbox{and}\quad
b_\star(k):=b_\star\e^{i(k+\beta_\star-\delta_\star/2)}.
$$
Therefore, the spectrum of $U_\star$ is given by
$$
\sigma(U_\star)=\big\{\lambda_{\star,j}(k)\mid j=1,2,~k\in[0,2\pi)\big\}%,
$$
with $\lambda_{\star,j}(k)$ the solution of the characteristic equation
$$
\det\big(\widehat{U_\star}(k)-\lambda_{\star,j}(k)\big)=0,\quad j=1,2,~k\in[0,2\pi).
$$
In case (a), we obtain
$$
\lambda_{\star,1}(k)=i\e^{i\delta_\star/2}
\quad\hbox{and}\quad
\lambda_{\star,2}(k)=-i\e^{i\delta_\star/2}.
$$
In case (b), we obtain
$$
\lambda_{\star,j}(k)
=\e^{i\delta_\star/2}\big(\tau_\star(k)+i\;\!(-1)^{j-1}\eta_\star(k)\big),\quad j=1,2.
$$
Finally, in case (c) we obtain
$$
\lambda_{\star,1}(k)=\e^{i(k+\alpha_\star)}
\quad\hbox{and}\quad
\lambda_{\star,2}(k)=\e^{-i(k+\alpha_\star-\delta_\star)}.
$$
\end{proof}

We now exhibit normalised eigenvectors $u_{\star,j}(k)$ of $\widehat{U_\star}(k)$
associated with the eigenvalues $\lambda_{\star,j}(k)$ which are $C^\infty$ in the
variable $k:$
$$
\begin{cases}
u_{\star,j}(k):=\frac{\sqrt{\eta_\star(k)+(-1)^{j-1}\varsigma_\star(k)}}
{b_\star\sqrt{2\eta_\star(k)}}
\begin{pmatrix}
ib_\star(k)\\
\varsigma_\star(k)+(-1)^j\eta_\star(k)
\end{pmatrix}
& \hbox{if $a_\star\in[0,1)$}\medskip\\
u_{\star,1}(k):=\begin{pmatrix}1\\0\end{pmatrix}
\quad\hbox{and}\quad
u_{\star,2}(k):=\begin{pmatrix}0\\1\end{pmatrix}
& \hbox{if $a_\star=1$.}
\end{cases}
$$
We leave the reader check that $u_{\star,j}(k)$ are indeed normalised eigenvectors of
$\widehat{U_\star}(k)$ with eigenvalues $\lambda_{\star,j}(k)$. In addition, since for
$a_\star\in[0,1)$ one has $\eta_\star(k)>0$ and
$\eta_\star(k)+(-1)^{j-1}\varsigma_\star(k)>0$, we note that the $2\pi$-periodic map
$\R\ni k\mapsto u_{\star,j}(k)\in\C^2$ is of class $C^\infty$.

Our next goal is to construct a suitable conjugate operator for the operator
$U_\star$. For this, a few preliminaries are necessary. First, we equip the interval
$[0,2\pi)$ with the addition modulo $2\pi$, and for any $n\in\N$ we define the space
$C^n\big([0,2\pi),\C^2\big)\subset\KK$ as the set of functions $[0,2\pi)\to\C^2$ of
class $C^n$. In particular, we have $u_{\star, j}\in C^\infty\big([0,2\pi),\C^2\big)$,
and the space $\F\H_{\rm fin}\subset C^\infty\big([0,2\pi),\C^2\big)$ is the set of
$\C^2$-valued trigonometric polynomials.

Next, we define the asymptotic velocity operator for the operator $U_\star$. For
$j=1,2$, we let $v_{\star,j}:[0,2\pi)\to\R$ be the bounded function given by
\begin{equation}\label{def_small_v}
v_{\star,j}(k)
:=\frac{i\;\!\lambda_{\star,j}'(k)}{\lambda_{\star,j}(k)},\quad k\in[0,2\pi).
\end{equation}
Here, $(\;\!\cdot\;\!)'$ stands for the derivative with respect to $k$, and
$v_{\star,j}$ is real valued because $\lambda_{\star,j}$ takes values in the complex
unit circle. Finally, for all $f\in\KK$ and almost every $k\in[0,2\pi)$, we define the
decomposable operator $\widehat{V_\star}\in\B(\KK)$ by
\begin{equation}\label{def_big_V}
\big(\widehat{V_\star}f\big)(k):=\widehat{V_\star}(k)f(k)
\quad\hbox{where}\quad
\widehat{V_\star}(k):=\sum_{j=1}^2v_{\star,j}(k)\;\!\Pi_{\star,j}(k)\in\B(\C^2),
\end{equation}
and we call asymptotic velocity operator the operator $V_\star$ given as inverse
Fourier transform of $\widehat{V_\star}$, namely,
$$
V_\star:=\F^*\;\!\widehat{V_\star}\;\!\F.
$$

The basic spectral properties of $V_\star$ are collected in the following lemma.

\begin{Lemma}[Spectrum of $V_\star$]
Let $C_\star$ be parameterised as in \eqref{eq_param_C}.
\begin{enumerate}
\item[(a)] If $a_\star=0$, then $v_{\star,j}=0$ for $j=1,2$, and $V_\star=0$.
\item[(b)] If $a_\star\in(0,1)$, then
$v_{\star,j}(k)=\frac{(-1)^{j}\varsigma_\star(k)}{\eta_\star(k)}$ for $j=1,2$ and
$k\in[0,2\pi)$, $\sigma_{\rm p}(V_\star)=\varnothing$ and
$$
\sigma(V_\star)=\sigma_{\rm c}(V_\star)=[-a_\star, a_\star].
$$
\item[(c)] If $a_\star=1$, then $v_{\star,j}=(-1)^{j}$ for $j=1,2$, and $V_\star$
has pure point spectrum
$$
\sigma(V_\star)=\sigma_{\rm p}(V_\star)=\{-1,1\}
$$
with each point an eigenvalue of $V_\star$ of infinite multiplicity.
\end{enumerate}
\end{Lemma}

\begin{proof}
The claims follow from simple calculations using the formulas for
$\lambda_{\star,j}(k)$ in the proof of Lemma \ref{lemma_spectrum_U_star} and the
definition \eqref{def_small_v} of $v_{\star,j}(k)$.
\end{proof}

For any $\xi,\zeta\in C\big([0,2\pi),\C^2\big)$, we define the operator
$|\xi\rangle\langle\zeta|:C\big([0,2\pi),\C^2\big)\to C\big([0,2\pi),\C^2\big)$
by
$$
\big(|\xi\rangle\langle\zeta|f\big)(k)
:=\big\langle\zeta(k),f(k)\big\rangle_2\;\!\xi(k),
\quad f\in C\big([0,2\pi),\C^2\big),~k\in[0,2\pi),
$$
where $\langle\;\!\cdot\;\!,\;\!\cdot\;\rangle_2$ is the usual scalar product on
$\C^2$. This operator extends continuously to an element of $\B(\KK)$, with norm
satisfying the bound
\begin{equation}\label{eq_estimate}
\big\||\xi\rangle\langle\zeta|\big\|_{\B(\KK)}
\le\|\xi\|_{\linf([0,2\pi),\frac{\d k}{2\pi},\C^2)}
\;\!\|\zeta\|_{\linf([0,2\pi),\frac{\d k}{2\pi},\C^2)}.
\end{equation}
We also define the self-adjoint operator $P$ in $\KK$
$$
Pf:=-if',\quad f\in\dom(P):=\big\{f\in\KK\mid\hbox{$f$ is absolutely continuous,
$f'\in\KK$, and $f(0)=f(2\pi)$}\big\}.
$$
With these definitions at hand, we can prove the self-adjointness of an operator
useful for the definition of our future the conjugate operator for $U:$

\begin{Lemma}\label{lemma_X}
The operator
$$
\widehat{X_\star}f
:=-\sum_{j=1}^2\big(\big|u_{\star,j}\big\rangle\big\langle u_{\star,j}\big|P
-i\;\!\big|u_{\star,j}\big\rangle\big\langle u_{\star,j}'\big|\big)f,
\quad f\in\F\H_{\rm fin},
$$
is essentially self-adjoint in $\KK$, with closure denoted by the same symbol. In
particular, the Fourier transform $X_\star:=\F^*\widehat{X_\star}\F$ of
$\widehat{X_\star}$ is essentially self-adjoint on $\H_{\rm fin}$ in $\H$.
\end{Lemma}

\begin{proof}
The proof consists in checking the assumptions of Nelson's commutator theorem
\cite[Thm.~X.37]{RS2} applied with the comparison operator $N:=P^2+1$.

For this, we first note that the operator $N$ is essentially self-adjoint on
$\F\H_{\rm fin}$ because it is the Fourier transform of a multiplication operator
acting on functions with finite support (see \cite[Ex.~5.1.15]{Ped89}). Next, by
performing an integration by parts with boundary terms canceling each other out, we
verify that $\widehat{X_\star}$ is symmetric on $\F\H_{\rm fin}$. Then, by using the
definition of $\widehat{X_\star}$ and the estimate \eqref{eq_estimate}, we check that
the inequality $\|\widehat{X_\star}f\|_\KK\le{\rm Const.}\;\!\|Nf\|_\KK$ holds for
each $f\in\F\H_{\rm fin}$. Finally, a direct calculation shows that for all
$\xi,\zeta\in C^2\big([0,2\pi),\C^2\big)$ and $f\in\F\H_{\rm fin}$
\begin{align*}
&\big\langle Nf,|\xi\rangle\langle\zeta|\;\!f\big\rangle_\KK
-\big\langle f,|\xi\rangle\langle\zeta|Nf\big\rangle_\KK\\
&=\big\langle f,\big(|\xi''\rangle\langle\zeta|
-|\xi\rangle\langle\zeta''|
-2\;\!|\xi'\rangle\langle\zeta'|-2i\;\!|\xi'\rangle\langle\zeta|\;\!P
-2i\;\!|\xi\rangle\langle\zeta'|\;\!P\big)f\big\rangle_\KK.
\end{align*}
This, together with the definition of $\widehat{X_\star}$, implies that
$$
\big|\big\langle\widehat{X_\star}f,Nf\big\rangle_\KK
-\big\langle Nf,\widehat{X_\star}f\big\rangle_\KK\big|
\le{\rm Const.}\;\!\langle f,Nf\rangle_\KK.
$$
Thus, all the assumptions of Nelson's commutator theorem are verified, and the claim
is proved.
\end{proof}

The main relations between the operators introduced so far are summarized in the
following proposition. To state it, we need one more decomposable operator
$\widehat{H_\star}\in\B(\KK)$ defined for all $f\in\KK$ and almost every
$k\in[0,2\pi)$ by
$$
\big(\widehat{H_\star}f\big)(k):=\widehat{H_\star}(k)f(k)
\quad\hbox{where}\quad
\widehat{H_\star}(k):=-\sum_{j=1}^2v_{\star,j}'(k)\;\!\Pi_{\star,j}(k)\in\B(\C^2).
$$
We also need the inverse Fourier transform $H_\star:=\F^*\widehat{H_\star}\F$ of
$\widehat{H_\star}$.

\begin{Proposition}\label{prop_tech}
\begin{enumerate}
\item[(a)] One has the equality $[iX_\star,V_\star]=H_\star$ in the form sense on
$\H_{\rm fin}$.
\item[(b)] $U_\star$, $V_\star$ and $H_\star$ are mutually commuting.
\item[(c)] One has the equality $[X_\star,U_\star]=U_\star V_\star$ in the form sense
on $\H_{\rm fin}$.
\end{enumerate}
\end{Proposition}	

\begin{proof}
(a) Let $f,g\in\F\H_{\rm fin}$. Then, a direct calculation using an integration by
parts (with boundary terms canceling each other out) implies that
$$
\big\langle\widehat{X_\star}f,i\;\!\widehat{V_\star}g\big\rangle_\KK
-\big\langle f,i\;\!\widehat{V_\star}\widehat{X_\star}g\big\rangle_\KK
=\big\langle f,\widehat{H_\star}g\big\rangle_\KK.
$$
Therefore, the claim follows by an application of the Fourier transform $\F$.

(b) The mutual commutativity of the operators $U_\star$, $V_\star$ and $H_\star$ is a
direct consequence of their boundedness and their definition in terms of the
orthogonal projections $\Pi_{\star,j}(k)$, $k\in[0,2\pi)$.

(c) As in point (a), the proof consists in computing for $f,g\in\F\H_{\rm fin}$ the
difference
$$
\big\langle\widehat{X_\star}f,\widehat{U_\star}g\big\rangle_\KK
-\big\langle f,\widehat{U_\star}\widehat{X_\star}g\big\rangle_\KK
$$
with an integration by parts, checking that this difference is equal to
$\big\langle g,\widehat{U_\star} \widehat{V_\star}f\big\rangle_\KK$, and applying the
Fourier transform $\F$.
\end{proof}

Since $X_\star$ is essentially self-adjoint on $\H_{\rm fin}$, Proposition
\ref{prop_tech}(a) implies that $V_\star\in C^1(X_\star)$. Therefore, the operator
$$
A_\star\Psi:=\frac12\big(X_\star V_\star+V_\star X_\star\big)\Psi,
\quad\Psi\in\dom(A_\star):=\big\{\Psi\in\H\mid V_\star\Psi\in\dom(X_\star)\big\},
$$
is self-adjoint in $\H$, and essentially self-adjoint on $\H_{\rm fin}$ (see
\cite[Lemma~2.4]{Tie16_2}). We can now state and prove the main results of this
section. We recall that $\Int(\Theta)$ and $\partial\Theta$ denote the interior and
the boundary of a set $\Theta\subset\T$. We also recall that the functions
$\varrho^{A_\star}_{U_\star}$ and $\widetilde{\varrho}^{A_\star}_{U_\star}$ have been
defined in Section \ref{Sec_one_Hilbert}.

\begin{Proposition}\label{prop_functions_rho}
\begin{enumerate}
\item[(a)] $U_\star\in C^1(A_\star)$ with $U_\star^{-1}[A_\star,U_\star]=V_\star^2$.
\item[(b)] $\varrho^{A_\star}_{U_\star}=\widetilde{\varrho}^{A_\star}_{U_\star}$, and
\begin{enumerate}
\item[(i)] if $a_\star=0$, then $\widetilde{\varrho}^{A_\star}_{U_\star}(\theta)=0$
for $\theta\in\big\{i\e^{i\delta_\star/2},-i\e^{i\delta_\star/2}\big\}$ and
$\widetilde{\varrho}^{A_\star}_{U_\star}(\theta)=\infty$ otherwise,
\item[(ii)] if $a_\star\in(0,1)$, then
$\widetilde{\varrho}^{A_\star}_{U_\star}(\theta)>0$ for
$\theta\in\Int\big(\sigma(U_\star)\big)$,
$\widetilde{\varrho}^{A_\star}_{U_\star}(\theta)=0$ for
$\theta\in\partial\sigma(U_\star)$, and
$\widetilde{\varrho}^{A_\star}_{U_\star}(\theta)=\infty$ otherwise,
\item[(iii)] if $a_\star=1$, then $\widetilde{\varrho}^{A_\star}_{U_\star}(\theta)=1$
for all $\theta\in\T$.
\end{enumerate}
\item[(c)]
\begin{enumerate}
\item[(i)] If $a_\star\in(0,1)$, then $U_\star$ has purely absolutely continuous
spectrum
$$
\sigma(U_\star)
=\sigma_{\rm ac}(U_\star)
=\big\{\e^{i\gamma}\mid
\gamma\in[\delta_\star/2+\theta_\star,\pi+\delta_\star/2-\theta_\star]
\cup[\pi+\delta_\star/2+\theta_\star,2\pi+\delta_\star/2-\theta_\star]\big\}.
$$
\item[(ii)] If $a_\star=1$, then $U_\star$ has purely absolutely continuous
spectrum $\sigma(U_\star)=\sigma_{\rm ac}(U_\star)=\T$.
\end{enumerate}
\end{enumerate}
\end{Proposition}

\begin{proof}
(a) A calculation in the forme sense on $\H_{\rm fin}$ using points (b) and (c) of
Proposition \ref{prop_tech} gives
$$
[A_\star,U_\star]
=\frac12\big(V_\star[X_\star,U_\star]+[X_\star,U_\star]V_\star\big)
=U_\star V_\star^2.
$$
Since $A_\star$ is essentially self-adjoint on $\H_{\rm fin}$, this implies that
$U_\star\in C^1(A_\star)$ with $U_\star^{-1}[A_\star,U_\star]=V_\star^2$.

(b) Take $\theta\in\T$ and $\varepsilon>0$. Then, using the result of point (a) and
\eqref{def_big_V}, we obtain for almost every $k\in[0,2\pi)$
\begin{align*}
\big(\F E^{U_\star}(\theta;\varepsilon)U_\star^{-1}[A_\star,U_\star]
E^{U_\star}(\theta;\varepsilon)\F^*\big)(k)
&=\big(\F E^{U_\star}(\theta;\varepsilon)\;\!V_\star^2
E^{U_\star}(\theta;\varepsilon)\F^*\big)(k)\\
&=E^{\widehat{U_\star}(k)}(\theta;\varepsilon)\;\!\widehat{V_\star}(k)^2
E^{\widehat{U_\star}(k)}(\theta;\varepsilon)\\
&\ge\min\big\{v_{\star,1}(k)^2,v_{\star,2}(k)^2\big\}
E^{\widehat{U_\star}(k)}(\theta;\varepsilon).
\end{align*}
Then, the definition \eqref{def_small_v} of $v_{\star,j}(k)$ shows that
$v_{\star,j}(k)=0$ if and only if $\lambda_{\star,j}'(k)=0$, which occurs when
$\lambda_{\star,j}(k)\in\partial\sigma(U_\star)$. Therefore, one gets
$\varrho^{A_\star}_{U_\star}=\widetilde{\varrho}^{A_\star}_{U_\star}$ by Lemma
\ref{lemma_properties}(d), and to conclude one just has to take into account the form
of the boundary sets $\sigma(U_\star)$ given in Lemma \ref{lemma_spectrum_U_star}.

(c) We know from point (a) that $U_\star\in C^1(A_\star)$ with
$U_\star^{-1}[A_\star,U_\star]=V_\star^2$, and Proposition \ref{prop_tech}(a) implies
that $V_\star\in C^1(A_\star)$. Thus, $U_\star\in C^2(A_\star)$. Therefore, if
$a_\star\in(0,1)$, we infer from point (b.ii) and Theorem \ref{thm_spec_U} that
$U_\star$ has no singular continuous spectrum in $\Int\big(\sigma(U_\star)\big)$.
This, together with Lemma \ref{lemma_spectrum_U_star}(b), implies the claim in the
case $a_\star\in(0,1)$. The claim in the case $a_\star=1$ is proved in a similar way.
\end{proof}
%--------------------------------------------------------------------------------------
\subsection{Mourre estimate for $U$}\label{Sec_Mourre}
%--------------------------------------------------------------------------------------

In this section, we use the Mourre estimate for the asymptotic operators $U_\ell$ and
$U_{\rm r}$ to derive a Mourre estimate for $U$. To achieve this, we apply the
abstract construction introduced in Section \ref{Sec_two_Hilbert}, starting by
choosing $\H_0:=\H\oplus\H$ as second Hilbert space and $U_0:=U_\ell \oplus U_{\rm r}$
as second unitary operator in $\H_0$.

The spectral properties of $U_0$ are obtained as a consequence of Lemma
\ref{lemma_spectrum_U_star}(a), Proposition \ref{prop_functions_rho}(c) and the direct
sum decomposition of $U_0:$

\begin{Lemma}[Spectrum of $U_0$]\label{lemma_spec_U_0}
One has $\sigma(U_0)=\sigma(U_\ell)\cup\sigma(U_{\rm r})$ and
$\sigma_{\rm sc}(U_0)=\varnothing$. Furthermore,
\begin{enumerate}
\item[(a)] if $a_\ell=a_{\rm r}=0$, then $U_0$ has pure point spectrum
$$
\sigma(U_0)
=\sigma_{\rm p}(U_0)
=\sigma_{\rm p}(U_\ell)\cup\sigma_{\rm p}(U_{\rm r})
=\big\{i\e^{i\delta_\ell/2},-i\e^{i\delta_\ell/2},i\e^{i\delta_{\rm r}/2},
-i\e^{i\delta_{\rm r}/2}\big\}
$$
with each point an eigenvalue of $U_0$ of infinite multiplicity,
\item[(b)] if $a_\ell=0$ and $a_{\rm r}\in(0,1]$, then
$\sigma_{\rm ac}(U_0)=\sigma_{\rm ac}(U_{\rm r})$ with $\sigma_{\rm ac}(U_{\rm r})$ as
in Proposition \ref{prop_functions_rho}(c), and
$$
\sigma_{\rm p}(U_0)
=\sigma_{\rm p}(U_\ell)
=\big\{i\e^{i\delta_\ell/2},-i\e^{i\delta_\ell/2}\big\}
$$
with each point an eigenvalue of $U_0$ of infinite multiplicity,
\item[(c)] if $a_\ell\in(0,1]$ and $a_{\rm r}=0$, then
$\sigma_{\rm ac}(U_0)=\sigma_{\rm ac}(U_\ell)$ with $\sigma_{\rm ac}(U_\ell)$ as in
Proposition \ref{prop_functions_rho}(c), and
$$
\sigma_{\rm p}(U_0)
=\sigma_{\rm p}(U_{\rm r})
=\big\{i\e^{i\delta_{\rm r}/2},-i\e^{i\delta_{\rm r}/2}\big\}
$$
with each point an eigenvalue of $U_0$ of infinite multiplicity,
\item[(d)] if $a_\ell,a_{\rm r}\in(0,1]$, then $U_0$ has purely absolutely continuous
spectrum
$$
\sigma(U_0)
=\sigma_{\rm ac}(U_0)
=\sigma_{\rm ac}(U_\ell)\cup\sigma_{\rm ac}(U_{\rm r})
$$
with $\sigma_{\rm ac}(U_\ell)$ and $\sigma_{\rm ac}(U_{\rm r})$ as in Proposition \ref{prop_functions_rho}(c).
\end{enumerate}
\end{Lemma}

Also, as intuition suggests and as already stated in Theorem \ref{thm_essential}, the
spectrum of $U_0$ coincides with the essential spectrum of $U$, namely,
$$
\sigma_{\rm ess}(U)=\sigma(U_\ell)\cup\sigma(U_{\rm r})=\sigma(U_0).
$$

\begin{proof}[Proof of Theorem \ref{thm_essential}]
The proof is based on an argument using crossed product $C^*$-algebras inspired from
\cite{GI02,Man02}.

Let $\A$ be the algebra of functions $\Z\to\B(\C^2)$ admitting limits at $\pm\infty$,
and let $\A_0$ be the ideal of $\A$ consisting in functions $\Z\to\B(\C^2)$ vanishing
at $\pm\infty$. Since $\A$ is equipped with an action of $\Z$ by translation, namely,
$$
\big(T_y\varphi\big)(x):=\varphi(x+y),\quad x,y\in\Z,~\varphi\in\A,
$$
we can consider the crossed product algebra $\A\rtimes\Z$, and the functoriality of
the crossed product implies the identities
\begin{equation}\label{eq_alg}
(\A\rtimes\Z)/(\A_0\rtimes\Z)
\cong(\A/\A_0)\rtimes\Z
=\big(\B(\C^2)\oplus\B(\C^2)\big)\rtimes\Z
=\big(\B(\C^2)\rtimes\Z\big)\oplus\big(\B(\C^2)\rtimes\Z\big),
\end{equation}
where the equality $\A/\A_0=\B(\C^2)\oplus\B(\C^2)$ is obtained by evaluation of the
functions $\varphi\in\A$ at $\pm\infty$.

Now, the algebras $\A\rtimes\Z$ and $\A_0\rtimes\Z$ can be faithfully represented in
$\H$ by mapping the elements of $\A$ and $\A_0$ to multiplication operators in $\H$
and the elements of $\Z$ to the shifts $T_z$. Writing $\mathfrak A$ and
$\mathfrak A_0$ for these representations of $\A\rtimes\Z$ and $\A_0\rtimes\Z$ in
$\H$, we can note three facts. First, $\mathfrak A_0$ is equal to the ideal of compact
operators $\K(\H)$. Secondly, the operator $U$ belongs to $\mathfrak A$, since
$$
U=SC
=T_1\begin{pmatrix}1&0\\0&0\end{pmatrix}C
+T_{-1}\begin{pmatrix}0&0\\0&1\end{pmatrix}C
$$
with $T_1,T_{-1}$ shifts and
$
\left(\begin{smallmatrix}1&0\\0&0\end{smallmatrix}\right)C,
\left(\begin{smallmatrix}0&0\\0&1\end{smallmatrix}\right)C
$
multiplication operators in $\H$. Thirdly, the essential spectrum of $U$ in
$\mathfrak A$ is equal to the spectrum of the image of $U$ in the quotient algebra
$\mathfrak A/\K(\H)=\mathfrak A/\mathfrak A_0$. These facts, together with
\eqref{eq_alg} and Lemma \ref{lemma_spec_U_0}, imply the equalities
$$
\sigma_{\rm ess}(U)
=\sigma \big(SC(-\infty)\oplus SC(+\infty)\big)
=\sigma\big(SC_\ell\oplus SC_{\rm r}\big)
=\sigma(U_\ell)\cup\sigma(U_{\rm r})
=\sigma(U_0),
$$
which prove the claim.
\end{proof}

Next, we define the identification operator $J\in\B(\H_0,\H)$ by
$$
J(\Psi_\ell,\Psi_{\rm r}):=j_\ell\;\!\Psi_\ell+j_{\rm r}\;\!\Psi_{\rm r},
\quad(\Psi_\ell,\Psi_{\rm r})\in\H_0,
$$
where
$$
j_{\rm r}(x):=
\begin{cases}
1 & \hbox{if $x\ge0$}\\
0 & \hbox{if $x\le-1$}
\end{cases}
\quad\hbox{and}\quad
j_\ell:=1-j_{\rm r}.
$$	
The adjoint operator $J^*\in\B(\H,\H_0)$ satisfies
$$
J^*\Psi=(j_\ell\;\!\Psi,j_{\rm r}\;\!\Psi),\quad\Psi\in\H.
$$
Moreover, using the same notation for the functions $j_\ell,j_{\rm r}$ and the
associated multiplication operators in $\H$, one directly gets:

\begin{Lemma}\label{lemma_J}
$J^*J=j_\ell \oplus j_{\rm r}$ is an orthogonal projection on $\H_0$, and
$JJ^*=1_\H$.
\end{Lemma}

The first result of the next lemma is an analogue of Proposition
\ref{prop_functions_rho}(a) in the Hilbert space $\H_0$. To state it, we need to
introduce the operator $A_0:=A_\ell\oplus A_{\rm r}$ (which will be used as a
conjugate operator for $U_0$) and the operator $V_0:=V_\ell\oplus V_{\rm r}$.

\begin{Lemma}\label{lemma_B_compact}
\begin{enumerate}
\item[(a)] $U_0\in C^1(A_0)$ with $U_0^{-1}[A_0, U_0]=V_0^2$.
\item[(b)] $B:=JU_0-UJ\in\K(\H_0,\H)$ and $B_*:=JU_0^*-U^*J\in\K(\H_0,\H)$.
\end{enumerate}
\end{Lemma}

\begin{proof}
The proof of point (a) is similar to the proof of Proposition
\ref{prop_functions_rho}(a); one just has to replace the operators
$U_\star,A_\star,V_\star$ in $\H$ by the operators $U_0,A_0,V_0$ in $\H_0$. For point
(b), a direct computation with $(\Psi_\ell,\Psi_{\rm r})\in\H_0$ gives
\begin{align}
B (\Psi_\ell,\Psi_{\rm r})
&=\big(j_\ell\;\!U_\ell\Psi_\ell+j_{\rm r}\;\!U_{\rm r}\Psi_{\rm r}\big)
-U\big(j_\ell\Psi_\ell+j_{\rm r}\Psi_{\rm r}\big)\nonumber\\
&=\big([j_\ell,U_\ell]-(U-U_\ell)\;\!j_\ell\big)\Psi_\ell
+\big([j_{\rm r},U_{\rm r}]-(U-U_{\rm r})\;\!j_{\rm r}\big)\Psi_{\rm r}\nonumber\\
&=\big([j_\ell,S]C_\ell-S(C-C_\ell)\;\!j_\ell\big)\Psi_\ell
+\big([j_{\rm r},S]C_{\rm r}-S(C-C_{\rm r})\;\!j_{\rm r}\big)\Psi_{\rm r}.\label{eq_B}
\end{align}
Since we have $[j_\star,S]\in\K(\H)$ and $(C-C_\star)\;\!j_\star\in\K(\H)$ as a
consequence of Assumption \ref{ass_short}, it follows that $B\in\K(\H_0,\H)$. The
inclusion $B_*\in\K(\H_0,\H)$ is proved in a similar way.
\end{proof}

The next step is to define a conjugate operator $A$ for $U$ by using the conjugate
operator $A_0$ for $U_0$. For this, we consider the operator $JA_0J^*$ which is
well-defined and symmetric on $\H_{\rm fin}$. We have the equality
\begin{equation}\label{eq_J_A_0_J_star}
JA_0J^*=j_\ell\;\!A_\ell\;\!j_\ell+j_{\rm r}\;\!A_{\rm r}\;\!j_{\rm r}
\quad\hbox{on}\quad\H_{\rm fin},
\end{equation}
and $JA_0J^*$ is essentially self-adjoint on $\H_{\rm fin}:$

\begin{Lemma}[Conjugate operator for $U$]\label{lemma_A_self}
The operator $JA_0J^*$ is essentially self-adjoint on $\H_{\rm fin}$, with
corresponding self-adjoint extension denoted by $A$.
\end{Lemma}

\begin{proof}
The operator $\widehat{j_\star}:=\F j_\star\F^*\in\B(\KK)$ satisfies
$\widehat{j_\star}\;\!\dom(P)\subset\dom(P)$ and $[\widehat{j_\star},P]=0$ on
$\dom(P)$. Therefore, we have the following equalities on $\F\H_{\rm fin}$
\begin{align*}
\F j_\star A_\star\;\!j_\star\F^*
&=\tfrac12\F j_\star\big(X_\star V_\star+V_\star X_\star\big)\;\!j_\star\F^*\\
&=\tfrac12\;\!\widehat{j_\star}\big(\widehat{X_\star}\widehat{V_\star}
+\widehat{V_\star}\widehat{X_\star}\big)\;\!\widehat{j_\star}\\
&=\widehat{j_\star}\big(\widehat{V_\star}\widehat{X_\star}
-\tfrac{i}2\widehat{H_\star}\big)\;\!\widehat{j_\star}\\
&=-\sum_{j=1}^2\Big(\widehat{j_\star}\big|v_{\star,j}u_{\star,j}\big\rangle
\big\langle u_{\star,j}\big|\;\!\widehat{j_\star}\;\!P
-i\;\!\widehat{j_\star}\big|v_{\star,j}u_{\star,j}\big\rangle
\big\langle u_{\star,j}'\big|\;\!\widehat{j_\star}\Big)
-\tfrac{i}2\;\!\widehat{j_\star}\widehat{H_\star}\widehat{j_\star}.
\end{align*}
which give on $\F\H_{\rm fin}$
$$
\F JA_0J^*\F^*
=-\sum_{j=1}^2\sum_{\star\in\{\ell,{\rm r}\}}
\widehat{j_\star}\big|v_{\star,j}u_{\star,j}\big\rangle
\big\langle u_{\star,j}\big|\;\!\widehat{j_\star}\;\!P
+i\sum_{j=1}^2\sum_{\star\in\{\ell,{\rm r}\}}
\widehat{j_\star}\;\!\big|v_{\star,j}u_{\star,j}\big\rangle
\big\langle u_{\star,j}'\big|\;\!\widehat{j_\star}
-\tfrac{i}2\sum_{\star\in\{\ell,{\rm r}\}}
\widehat{j_\star}\widehat{H_\star}\widehat{j_\star}.
$$
The rest of the proof consists in an application of Nelson's commutator theorem
\cite[Thm.~X.37]{RS2} with the comparison operator $N:=P^2+1$. The estimates necessary
to apply the theorem are similar to the ones mentioned in the proof of Lemma
\ref{lemma_X}. As a consequence, it follows that $\F JA_0J^*\F^*$ is essentially
self-adjoint on $\F\H_{\rm fin}$, and thus that $JA_0J^*$ is essentially self-adjoint
on $\H_{\rm fin}$.
\end{proof}

We are thus in the setup of Assumption \ref{ass_eaa} with the set $\DD=\H_{\rm fin}$.
So, the next step is to show the inclusion $U\in C^1(A)$. For this, we use Corollary
\ref{Corol_C1(A)}. Using Corollary \ref{Corol_est_supp}, we also get an additional
compacity result:

\begin{Lemma}\label{lemma_last}
$U\in C^{1}(A)$ and $JU_0^{-1}[A_0, U_0]J^*-U^{-1}[A,U]\in\K(\H)$.
\end{Lemma}

\begin{proof}
First, we recall that $U_0\in C^1(A_0)$ due to Lemma \ref{lemma_B_compact}(a), and
that Assumption \ref{ass_eaa} holds with $\DD=\H_{\rm fin}$. Next, we note that the
expression for $B(\Psi_\ell,\Psi_{\rm r})$ with $(\Psi_\ell,\Psi_{\rm r})\in\H_0$ is
given in \eqref{eq_B}, and that
$$
B_*(\Psi_\ell,\Psi_{\rm r})
=\big(C^*[j_\ell,S^*]-(C^*-C_\ell^*)\;\!j_\ell S^*\big)\Psi_\ell
+\big(C^*[j_{\rm r},S^*]-(C^*-C_{\rm r}^*)\;\!j_{\rm r}S^*\big)\Psi_{\rm r}.
$$
Furthermore, we know from Lemma \ref{lemma_B_compact}(b) that $B,B_*\in\K(\H_0,\H)$.
In consequence, due to Corollaries \ref{Corol_C1(A)}-\ref{Corol_est_supp}, the claims
will follow if we show that $\overline{BA_0\upharpoonright\dom(A_0)}\in\B(\H_0,\H)$
and $\overline{B_*A_0\upharpoonright\dom(A_0)}\in\K(\H_0,\H)$. For this, we first note
that computations as in the proof of Lemma \ref{lemma_A_self} imply on $\H_{\rm fin}$
the equalities
\begin{align}
A_\star&=X_\star V_\star+\tfrac{i}2H_\star\nonumber\\
&=-\F^*\left\{P\sum_{j=1}^2\Big(\big|u_{\star,j}\big\rangle
\big\langle v_{\star,j}u_{\star,j}\big|+i\;\!\big|u'_{\star,j}\big\rangle
\big\langle v_{\star,j} u_{\star,j}\big|\Big)\right\}\F+\tfrac{i}2H_\star\nonumber\\
&=Q\;\!\F^*\left\{\sum_{j=1}^2\Big(\big|u_{\star,j}\big\rangle
\big\langle v_{\star,j} u_{\star,j}\big|+i\;\!\big|u'_{\star,j}\big\rangle\big\langle
v_{\star,j} u_{\star,j}\big|\Big)\right\}\F+\tfrac{i}2H_\star\label{eq_X_onleft}
\end{align}
with $Q$ the self-adjoint multiplication operator defined by
\begin{equation}\label{eq_def_Q}
\big(Q\;\!\Psi\big)(x)=x\;\!\Psi(x),
\quad x\in\Z,~\Psi\in\dom(Q):=\big\{\Psi\in\H\mid\|Q\;\!\Psi\|_\H<\infty\big\}.
\end{equation}
Therefore, since all the operators on the right of $Q$ in \eqref{eq_X_onleft} are
bounded, it is sufficient to show that
$$
\overline{B(Q\oplus Q)\upharpoonright\dom(Q)\oplus\dom(Q)}\in\B(\H_0,\H)
\quad\hbox{and}\quad
\overline{B_*(Q\oplus Q)\upharpoonright\dom(Q)\oplus\dom(Q)}\in\K(\H_0,\H).
$$
However, this can be deduced from the Assumption \ref{ass_short} once the following
observations are made: $\big[j_\star,S\big]=Sm_\star$ with $m_\star:\Z\to\B(\C^2)$ a
function with compact support, $[j_\star,S^*]=S^*n_\star$ with $n_\star:\Z\to\B(\C^2)$
a function with compact support, and $S^*Q=QS^*+b$ with
$b\in\linf\big(\Z,\B(\C^2)\big)$.
\end{proof}

We now recall that the set
$$
\tau(U):=\partial\sigma(U_\ell)\cup\partial\sigma(U_{\rm r}).
$$
has been introduced in Section \ref{Sec_model}. Due to Lemma
\ref{lemma_spectrum_U_star}, $\tau(U)$ contains at most $8$ values. Moreover, since we
show in the next proposition that a Mourre estimate holds on the set
$\{\sigma(U_\ell)\cup\sigma(U_{\rm r})\}\setminus\tau(U)$, it is natural to interpret
$\tau(U)$ as the set of thresholds in the spectrum of $U$.

\begin{Proposition}[Mourre estimate for $U$]\label{prop_Mourre}
We have $\widetilde\varrho_U^A\ge \widetilde\varrho_{U_0}^{A_0}$ with
$
\widetilde\varrho_{U_0}^{A_0}
=\min\big\{\widetilde\varrho_{U_\ell}^{A_\ell},
\widetilde\varrho_{U_{\rm r}}^{A_{\rm r}}\big\}
$
and $\widetilde\varrho_{U_\ell}^{A_\ell},\widetilde\varrho_{U_{\rm r}}^{A_{\rm r}}$
given in Proposition \ref{prop_functions_rho}. In particular,
$\widetilde\varrho_{U_0}^{A_0}(\theta)>0$ if
$\theta\in\{\sigma(U_\ell)\cup\sigma(U_{\rm r})\}\setminus\tau(U)$.
\end{Proposition}

\begin{proof}
The first claim follows from Theorem \ref{thm_fonctionrho}, with the assumptions of
this theorem verified in Lemmas \ref{lemma_J}-\ref{lemma_last}. The second claim
follows from Proposition \ref{prop_functions_rho} and standard results on the function
$\widetilde\varrho_{U_0}^{A_0}$ when $A_0$ and $U_0$ are direct sums of operators (see
\cite[Prop.~8.3.5]{ABG96} for a proof in the case of direct sums of self-adjoint
operators).
\end{proof}

%--------------------------------------------------------------------------------------
\subsection{Spectral properties of $U$}\label{Sec_spectrum}
%--------------------------------------------------------------------------------------

In order to go one step further in the study of $U$, a regularity property of $U$ with
respect to $A$ stronger than $U\in C^1(A)$ has to be established. This regularity
property will be obtained by considering first the operator $JU_0 J^*$, and then by
analysing the difference $U-JU_0J^*$. We note that $JU_0J^*$ and $U-JU_0J^*$ satisfy
the equalities
\begin{equation}\label{eq_J_U_0_J_star}
JU_0J^*=j_\ell\;\!U_\ell\;\!j_\ell+j_{\rm r}\;\!U_{\rm r}\;\!j_{\rm r}
\end{equation}
and
\begin{equation}\label{eq_difference}
U-JU_0J^*=j_\ell(U-U_\ell)\;\!j_\ell+j_{\rm r}(U-U_{\rm r})\;\!j_{\rm r}
+j_\ell\;\!U\;\!j_{\rm r}+j_{\rm r}\;\!U\;\!j_\ell.
\end{equation}

\begin{Lemma}\label{lemma_C2}
$JU_0J^*\in C^2(A)$.
\end{Lemma}

\begin{proof}
The proof is based on standard results from toroidal pseudodifferential calculus, as
presented for example in \cite[Chap.~4]{RT10_2}. The normalisation we use for the
Fourier transform differs from the one used in \cite{RT10_2}, but the difference is
harmless.

(i) First, we note that $\widehat{j_\star}$ is a toroidal pseudodifferential operator
on $\F\H_{\rm fin}$ with symbol in $S^0_{\rho,0}(\T\times\Z)$ for each $\rho>0$ (see
the definitions 4.1.7 and 4.1.9 of \cite{RT10_2} for details). Similarly, the equation
\eqref{eq_X_onleft} shows that $\widehat{A_\star}$ is a first order differential
operator on $\F\H_{\rm fin}$ with matrix coefficients in
$\M\big(2,C^\infty(\T)\big)\subset\M\big(2,S^0_{\rho,0}(\T\times\Z)\big)$ for each
$\rho>0$. In consequence, it follows from \cite[Thm.~4.7.10]{RT10_2} that the
commutator $\big[\widehat{j_\star},\widehat{A_\star}\big]$ on $\F\H_{\rm fin}$ is
well-defined and equal to a toroidal pseudodifferential operator with matrix
coefficients in $\M\big(2,S^{1-\rho}_{\rho,0}(\T\times\Z)\big)$ for each $\rho>0$.
This implies that $\big[\widehat{j_\star},\widehat{A_\star}\big]$ is bounded on
$\F\H_{\rm fin}$, and thus that $\widehat{j_\star}\in C^1(\widehat{A_\star})$ since
$\F\H_{\rm fin}$ is dense in $\dom(\widehat{A_\star})$. By Fourier transform, it
follows that $j_\star\in C^1(A_\star)$.

(ii) A calculation in the form sense on $\H_{\rm fin}$ using equations
\eqref{eq_J_A_0_J_star} and \eqref{eq_J_U_0_J_star} and the identities
$j_\ell\;\!j_{\rm r}=0=j_{\rm r}\;\!j_\ell$ gives
\begin{align}
\big[JU_0J^*,A\big]
&=\big[j_\ell\;\!U_\ell\;\!j_\ell,j_\ell\;\!A_\ell\;\!j_\ell\big]
+\big[j_{\rm r}\;\!U_{\rm r}\;\!j_{\rm r},j_{\rm r}\;\!A_{\rm r}\;\!j_{\rm r}\big]
\nonumber\\
&=\sum_{\star\in\{\ell,{\rm r}\}}j_\star\big(U_\star\;\!j_\star A_\star
-A_\star\;\!j_\star U_\star\big)\;\!j_\star\nonumber\\
&=\sum_{\star\in\{\ell,{\rm r}\}}j_\star\big(\big[U_\star,j_\star\big]A_\star
+\big[j_\star U_\star,A_\star\big]\big)\;\!j_\star.\label{eq_2_terms}
\end{align}
Since $j_\star U_\star\in C^1(A_\star)$ by Proposition \ref{prop_functions_rho}(a),
point (i) and \cite[Prop.~5.1.5]{ABG96}, the second term on the r.h.s. of
\eqref{eq_2_terms} belongs to $\B(\H)$. Furthermore, a calculation using the
definition of the shift operator $S$ shows that
$$
\big[U_\star,j_\star\big]
=\big[S,j_\star\big]C_\star
=B_\star m_\star
$$
with $B_\star\in\B(\H)$ and $m_\star:\Z\to\B(\C^2)$ a function with compact support.
It follows from \eqref{eq_X_onleft} that $\big[U_\star,j_\star\big]A_\star$ is bounded
on $\H_{\rm fin}$. Therefore, both terms on the r.h.s. of \eqref{eq_2_terms} are
bounded on $\H_{\rm fin}$, and thus we infer from the density of $\H_{\rm fin}$ in
$\dom(A)$ that $JU_0J^*\in C^1(A)$.

(iii) To show that $JU_0J^*\in C^2(A)$, one has to commute the r.h.s. of
\eqref{eq_2_terms} once more with $A$. Doing this in the form sense on $\H_{\rm fin}$
with the notation $\sum_{\star\in\{\ell,{\rm r}\}}j_\star D_\star\;\!j_\star$ with
$D_\star:=[U_\star,j_\star]A_\star+[j_\star U_\star,A_\star]$ for the r.h.s. of
\eqref{eq_2_terms}, one gets that $JU_0J^*\in C^2(A)$ if the operators
$[D_\star,A_\star]$, $[D_\star,j_\star]A_\star$ and $A_\star[D_\star,j_\star]$ defined
in the form sense on $\H_{\rm fin}$ extend continuously to elements of $\B(\H)$.

For the first operator, we have in the form sense on $\H_{\rm fin}$ the equalities
\begin{align}
[D_\star,A_\star]
&=\big[[U_\star,j_\star]A_\star+j_\star[U_\star,A_\star]
+[j_\star,A_\star]U_\star,A_\star\big]\nonumber\\
&=\big[[U_\star,j_\star]A_\star,A_\star\big]
+j_\star\big[[U_\star,A_\star],A_\star\big]
+[j_\star,A_\star][U_\star,A_\star]
+[j_\star,A_\star][U_\star,A_\star]
+\big[[j_\star,A_\star],A_\star\big]U_\star\;\!. \label{eq_first_com}
\end{align}
Then, simple adaptations of the arguments presented in points (i) and (ii) show that
the operators $[j_\star,A_\star],[U_\star,j_\star]\in\B(\H)$ can be multiplied in the
form sense on $\H_{\rm fin}$ by several operators $A_\star$ on the left and/or on the
right and that the resultant operators extend continuously to elements of $\B(\H)$.
Therefore, the first, the third, the fourth and the fifth terms in
\eqref{eq_first_com} extend continuously to elements of $\B(\H)$. For the second term,
we note from Propositions \ref{prop_tech}(a) and \ref{prop_functions_rho}(a) that
$U_\star,V_\star\in C^1(A_\star)$ with $[U_\star,A_\star]=-U_\star V_\star^2$. In
consequence, we have $U_\star V_\star^2\in C^1(A_\star)$ by \cite[Prop.~5.1.5]{ABG96}
and
$$
j_\star\big[[U_\star,A_\star],A_\star\big]
=-j_\star\big[U_\star V_\star^2,A_\star\big]
\in\B(\H).
$$

The proof that the operators $[D_\star,j_\star]A_\star$ and $A_\star[D_\star,j_\star]$
defined in the form sense on $\H_{\rm fin}$ extend continuously to elements of
$\B(\H)$ is similar. The only noticeable difference is the appearance of terms
$[U_\star V_\star^2,j_\star]A_\star$ and $A_\star [U_\star V_\star^2,j_\star]$.
However, by observing that $V_\star^2\in C^1(A_\star)$ and that $[V^2_\star,j_\star]$
is a toroidal pseudodifferential operator with matrix coefficients in
$\M\big(2,S^{-\rho}_{\rho,0}(\T\times\Z)\big)$ for each $\rho>0$, one infers that
$[U_\star V_\star^2,j_\star]A_\star$ and $A_\star [U_\star V_\star^2,j_\star]$ extend
continuously to elements of $\B(\H)$.
\end{proof}

In the next lemma, we prove that $U$ satisfies sufficient regularity with respect to
$A$, namely that $U\in C^{1+\varepsilon}(A)$ for some $\varepsilon\in(0,1)$. We recall
from Section \ref{Sec_one_Hilbert} that the sets $C^2(A)$, $C^{1+\varepsilon}(A)$,
$C^{1+0}(A)$ and $C^{1,1}(A)$ satisfy the continuous inclusions
$$
C^2(A)\subset C^{1+\varepsilon}(A)\subset C^{1+0}(A)\subset C^{1,1}(A).
$$

\begin{Lemma}\label{lemma_C_1_epsilon}
$U\in C^{1+\varepsilon}(A)$ for each $\varepsilon\in(0,1)$ with
$\varepsilon\le\min\{\varepsilon_\ell,\varepsilon_{\rm r}\}$.
\end{Lemma}

\begin{proof}
(i) Since $JU_0J^*\in C^2(A)$ by Lemma \ref{lemma_C2} and
$C^2(A)\subset C^{1+\varepsilon}(A)$, it is sufficient to show that
$U-JU_0J^*\in C^{1+\varepsilon}(A)$, with $U-JU_0J^*$ given by \eqref{eq_difference}.
Moreover, calculations as in the proof of Lemma \ref{lemma_C2} show that the last two
terms $j_\ell\;\!U\;\!j_{\rm r}$ and $j_{\rm r}\;\!U\;\!j_\ell$ of
\eqref{eq_difference} belong to $C^2(A)$. So, it only remains to show that
$
j_\ell(U-U_\ell)\;\!j_\ell+j_{\rm r}(U-U_{\rm r})\;\!j_{\rm r}
\in C^{1+\varepsilon}(A)
$.

(ii) In order to show the mentioned inclusion, we first observe from
\eqref{eq_def_U_ell} and \eqref{eq_J_A_0_J_star} that we have in the form sense on
$\H_{\rm fin}$ the equalities
\begin{align}
\big[j_\ell(U-U_\ell)\;\!j_\ell+j_{\rm r}(U-U_{\rm r})\;\!j_{\rm r},A\big]
&=\sum_{\star\in\{\ell,{\rm r}\}}
\big[j_\star(U-U_\star)\;\!j_\star,j_\star A_\star\;\!j_\star\big]\nonumber\\
&=\sum_{\star\in\{\ell,{\rm r}\}}\big(j_\star S(C-C_\star)
\;\!j_\star A_\star\;\!j_\star
-j_\star A_\star\;\!j_\star S(C-C_\star)\;\!j_\star\big).\label{eq_first_step}
\end{align}
Then, using Assumption \ref{ass_short}, the formula \eqref{eq_X_onleft} for $A_\star$
on $\H_{\rm fin}$, and a similar formula with the operator $Q$ on the right (recall
that $Q$ is the position operator defined in \eqref{eq_def_Q}), one obtains that the
operator on the r.h.s. of \eqref{eq_first_step} defined as
$$
D_\star:=j_\star S(C-C_\star)\;\!j_\star A_\star\;\!j_\star
-j_\star A_\star\;\!j_\star S(C-C_\star)\;\!j_\star
$$
extends continuously to an element of $\B(\H)$. Since $\H_{\rm fin}$ is dense in
$\dom(A)$, this implies that
$j_\ell(U-U_\ell)\;\!j_\ell+j_{\rm r}(U-U_{\rm r})\;\!j_{\rm r}\in C^1(A)$.

(iii) To show that
$
j_\ell(U-U_\ell)\;\!j_\ell+j_{\rm r} (U-U_{\rm r})\;\!j_{\rm r}
\in C^{1+\varepsilon}(A)
$,
it remains to check that
$$
\big\|\e^{-itA}D_\star\e^{itA}-D_\star\big\|_{\B(\H)}
\le{\rm Const.}\;\!t^\varepsilon\quad\hbox{for all $t\in(0,1)$.}
$$
But, algebraic manipulations as presented in \cite[p.~325-326]{ABG96} show that for
all $t\in(0,1)$
\begin{align*}
\big\|\e^{-itA}D_\star\e^{itA}-D_\star\big\|_{\B(\H)}
&\le{\rm Const.}\;\!\big(\|\sin(tA)D_\star\|_{\B(\H)}
+\|\sin(tA)(D_\star)^*\|_{\B(\H)}\big)\\
&\le{\rm Const.}\;\!\big(\|tA\;\!(tA+i)^{-1}D_\star\|_{\B(\H)}
+\|tA\;\!(tA+i)^{-1}(D_\star)^*\|_{\B(\H)}\big).
\end{align*}
Furthermore, if we set $A_t:=tA\;\!(tA+i)^{-1}$ and
$\Lambda_t:=t\langle Q\rangle(t\langle Q\rangle+i)^{-1}$, we obtain that
$$
A_t=\big(A_t+i(tA +i)^{-1}A\;\!\langle Q\rangle^{-1}\big)\Lambda_t
$$
with $A\langle Q\rangle^{-1}\in\B(\H)$ due to
\eqref{eq_J_A_0_J_star}-\eqref{eq_X_onleft}. Thus, since
$\big\|A_t+i(tA +i)^{-1}A\;\!\langle Q\rangle^{-1}\big\|_{\B(\H)}$ is bounded by a
constant independent of $t\in(0,1)$, it is sufficient to prove that
$$
\|\Lambda_t D_\star\|_{\B(\H)}+\|\Lambda_t(D_\star)^*\|_{\B(\H)}
\le{\rm Const.}\;\!t^\varepsilon\quad\hbox{for all $t\in(0,1)$.}
$$
Now, this estimate will hold if we show that the operators
$\langle Q\rangle^\varepsilon D_\star$ and $\langle Q\rangle^\varepsilon(D_\star)^*$
defined in the form sense on $\H_{\rm fin}$ extend continuously to elements of
$\B(\H)$. For this, we fix $\varepsilon\in(0,1)$ with
$\varepsilon\le\min\{\varepsilon_\ell,\varepsilon_{\rm r}\}$, and note that
$\langle Q\rangle^{1+\varepsilon}(C-C_\star)\;\!j_\star\in\B(\H)$. With this inclusion
and the fact that $\langle Q\rangle^{-1}A_\star$ defined in the form sense on
$\H_{\rm fin}$ extend continuously to elements of $\B(\H)$, one readily obtains that
$\langle Q\rangle^\varepsilon D_\star$ and $\langle Q\rangle^\varepsilon(D_\star)^*$
defined in the form sense on $\H_{\rm fin}$ extend continuously to elements of
$\B(\H)$, as desired.
\end{proof}

With what precedes, we can now prove our last two main results on $U$ which have been
stated in Section \ref{Sec_model}.

\begin{proof}[Proof of Theorem \ref{thm_U_smooth_walks}]
Theorem \ref{thm_U_smooth}, whose assumptions are verified in Proposition
\ref{prop_Mourre} and Lemma \ref{lemma_C_1_epsilon}, implies that each $T\in\B(\H,\G)$
which extends continuously to an element of
$\B\big(\dom(\langle A\rangle^s)^*,\G\big)$ for some $s>1/2$ is locally $U$-smooth on
$\Theta\setminus\sigma_{\rm p}(U)$. Moreover, we know from the proof of of Lemma
\ref{lemma_C_1_epsilon} that $\dom(Q)\subset\dom(A)$. Therefore, we have
$\dom(\langle Q\rangle^s)\subset\dom(\langle A\rangle^s)$ for each $s>1/2$, and it
follows by duality that
$
\dom(\langle A\rangle^s)^*
\subset\dom(\langle Q\rangle^s)^*
\equiv\dom(\langle Q\rangle^{-s})
$
for each $s>1/2$. In consequence, any operator $T\in\B(\H,\G)$ which extends
continuously to an element of $\B\big(\dom(\langle Q\rangle^{-s}),\G\big)$ some
$s>1/2$ also extends continuously to an element of
$\B\big(\dom(\langle A\rangle^s)^*,\G\big)$. This concludes the proof.
\end{proof}

\begin{proof}[Proof of Theorem \ref{thm_spectrum_U_walks}]
The claim follows from Theorem \ref{thm_spec_U}, whose hypotheses are verified in
Lemma \ref{lemma_C_1_epsilon} and Proposition \ref{prop_Mourre}.
\end{proof}

%--------------------------------------------------------------------------------------
%\bibliography{../bibliographie/bibliographie}
%--------------------------------------------------------------------------------------

%--------------------------------------------------------------------------------------

\end{document}